\documentclass[journal]{IEEEtran}

\usepackage{url}
\usepackage[utf8]{inputenc} \usepackage[nolist,nohyperlinks]{acronym}

\usepackage[ruled,vlined]{algorithm2e}
\usepackage{amsmath}
\usepackage{amsfonts}
\usepackage{amssymb}
\usepackage{amsthm}
\usepackage{dsfont}
\usepackage{array}
\usepackage{multirow}
\usepackage{lipsum}
\usepackage{mathtools}
\usepackage{cuted}
\usepackage[caption=false]{subfig} 
\usepackage{bbm}
\usepackage{siunitx}
\usepackage{bm}
\usepackage{booktabs}

\usepackage{xcolor}

\usepackage{stfloats}
\usepackage{rotating} 
\usepackage{cases}

\usepackage{lineno}
\DeclareMathOperator{\EX}{\mathbb{E}}

\newcommand{\setS}{\mathcal{S}}

\newcommand{\setK}{\mathcal{K}}

\newcommand{\setT}{\mathcal{T}}

\newcommand{\setD}{\mathcal{D}}

\newcommand{\setH}{\mathcal{H}}
\newcommand{\setE}{\mathcal{E}}

\newcommand{\bx}{\bm{x}}
\newcommand{\bp}{\bm{p}}

\newcommand{\bw}{\vec{w}}

\newcommand{\bg}{\vec{g}}
\newcommand{\by}{\vec{y}}
\newcommand{\bn}{\vec{n}}

\newcommand{\estimG}{\tilde{\bg}}

\DeclareMathOperator{\Var}{Var}

\newcommand{\MSEOTA}{\ensuremath{\mathrm{MSE}_{\mathrm{OTA}}}}
\newcommand{\tMSEOTA}{\ensuremath{\widetilde{\mathrm{MSE}}_{\mathrm{OTA}}}}

\DeclareMathOperator{\Prob}{\mathbb{P}}

\DeclareSIUnit \belm {Bm}

\newcommand{\R}{\ensuremath{\mathbb{R}}}  
\newcommand{\C}{\ensuremath{\mathbb{C}}}  

\renewcommand{\vec}[1]{\ensuremath{\bm{\MakeLowercase{#1}}}}

\newcommand{\defeq}{\ensuremath{\triangleq}}

\usepackage{dsfont}
\newcommand{\Indic}{\mathds{1}}

\usepackage{enumitem}
\usepackage{cleveref}
\newtheorem{proposition}{Proposition}
\newtheorem{assumption}{Assumption}
\newtheorem{corollary}{Corollary}
\newtheorem{remark}{Remark}

\crefname{assumption}{assumption}{assumptions}
\Crefname{assumption}{Assumption}{Assumptions}
\crefname{figure}{fig.}{figs.}
\Crefname{figure}{Fig.}{Figs.}
\newtheorem{theorem}{Theorem}
\newtheorem{lemma}{Lemma}

\newcommand\blfootnote[1]{\begingroup
  \renewcommand\thefootnote{}\footnote{#1}\addtocounter{footnote}{-1}\endgroup
}

\begin{document}
\bstctlcite{MyBSTcontrol}
\title{
Analytically Characterized Optimal Power Control for 
Signal-Level-Integrated Sensing, Computing and Communication in Federated Learning
}

\author{
Paul Zheng,
Yao Zhu, 
Xiaopeng Yuan,
Yulin Hu, and
Anke Schmeink
\thanks{The work of Y. Hu and Y. Zhu was supported in part by the National Key R\&D Program of China under Grant 2023YFE0206600, NSFC Grant 62471341, and by the Hubei Provincial Science and Technology Cooperation Project under Grant 2025EHA040.
        The work of P. Zheng, X. Yuan, and A.
Schmeink is supported by the Federal Ministry of Research, Technology and Space (BMFTR) Germany projects 6GEM+ under Grant 16KIS2409K and GEM-X under Grant 16KISS004K and by German Research Council (DFG) Project under Grant SCHM 2643/27-1.}
\thanks{Part of this work has been presented in IEEE ICC 2024~\cite{Zheng_ICCwkp_2024}.}
\thanks{
P. Zheng, X. Yuan, and A. Schmeink are with Chair of Information Theory
and Data Analytics, RWTH Aachen University, Germany. (Email:~$zheng|yuan|schmeink$@inda.rwth-aachen.de)
}
\thanks{
Y. Zhu, and Y. Hu are with School of Electronic Information,
Wuhan University, 430072 Wuhan, China
(Email: $yao.zhu|yulin.hu$@whu.edu.cn)
}

}

\IEEEoverridecommandlockouts
\IEEEpubid{\begin{minipage}{\textwidth}\ \\[50pt]
		{\copyright 2026 IEEE. Personal use of this material is permitted.  Permission from IEEE must be obtained for all other uses, in any current or future media, including reprinting/republishing this material for advertising or promotional purposes, creating new collective works, for resale or redistribution to servers or lists, or reuse of any copyrighted component of this work in other works. Citation information: DOI: 10.1109/TWC.2026.3747602}
\end{minipage}} 

\maketitle

\blfootnote{Y. Zhu and Y. Hu are the corresponding authors.}
\begin{abstract}
In next-generation wireless edge systems, efficient functionality integration is essential to address the growing demands of communication, computation, and sensing. Signal-level integrated sensing, computing, and communication (Sig-ISCC) is envisioned, where a single waveform simultaneously supports sensing, computing and communication via over-the-air computation (AirComp).
Meanwhile, federated learning (FL) is widely regarded as a promising distributed machine learning framework that enables network intelligence in a privacy-preserving and secure manner, and exhibits strong synergy with AirComp, which alleviates the communication bottleneck of FL.
In this paper, we study uplink Sig-ISCC design for AirComp-FL with joint target detection.
We formulate the joint power and receive-scaling control problem, where edge devices' transmitted signals should serve both sensing and AirComp purposes. The goal is to minimize the AirComp aggregation distortion subject to a joint target-detection requirement.
Although the resulting problem is non-convex in the original variables, we show that it admits an equivalent convex reformulation after a suitable variable transformation.
By exploiting analytical optimality properties, we develop a robust, optimal, and polynomial-time-complexity algorithm that efficiently achieves the optimal transmit powers and receive scaling factor.
Simulation results validate the optimality and numerical robustness of the proposed algorithm and show its superior FL performance compared to baseline methods.

\end{abstract} 


\IEEEpeerreviewmaketitle

\acrodef{DEL}{distributed edge learning}
\acrodef{FL}{federated learning}
\acrodef{CR}{communication round}
\acrodef{AirComp}{over-the-air computation}
\acrodef{AirFEEL}{over-the-air federated edge learning}
\acrodef{ISAC}{integrated sensing and communication}
\acrodef{ISCC}{integrated sensing, computing and communication}
\acrodef{Sig-ISCC}{signal-level \ac{ISCC}}
\acrodef{BS}{base station}
\acrodef{MSE}{mean squared error}
\acrodef{MINLP}{mixed-integer nonlinear programming}
\acrodef{MEC}{mobile edge computing}
\acrodef{RIS}{reconfigurable intelligent surface}
\acrodef{ED}{edge device}
\acrodef{CSI}{channel state information}
\acrodef{NLP}{non-linear programming}
\acrodef{SGD}{stochastic gradient descent}

\section{Introduction}

\Ac{DEL} paradigms such as \ac{FL}~\cite{mcmahan_FL_2017} are becoming increasingly important owing to their ability to preserve data privacy and reduce communication costs~\cite{survey_DL2FL_2022}. \Ac{FL} trains a global model by repeatedly exchanging locally trained model updates (instead of raw data) between \acp{ED} and a parameter server. It proceeds iteratively in the form of \acp{CR}: 
in each \ac{CR}, a subset of selected \acp{ED} perform local training and upload their updates to the server for aggregation, i.e., simple (weighted) averaging in the original FedAvg~\cite{mcmahan_FL_2017}.
This framework has been considered for, and in some cases already deployed in various real-world mobile edge applications~\cite{Li2020_reviewFL_applications}, e.g., Google keyboard prediction~\cite{hard_2019_FLmobilekeyboard}. 
Since most \acp{ED} are wirelessly connected and frequent exchange of model updates is required, wireless communication constitutes a crucial bottleneck for such \ac{DEL} systems~\cite{Chen_DistributedlearningWireless_survey_2021}.

Since the server only requires the average of the local updates, and neural networks are robust to small perturbations of model parameters, analog transmission leveraging \ac{AirComp} has emerged as an effective technique to address the communication bottleneck of \ac{FL}, yielding what is known as \ac{AirFEEL}~\cite{yang_FL_OTA_2020,Amiri_OTA-FL_2020,zhu_OTA-FL_2020}. \Ac{AirComp} exploits the superposition property of multi-access channels: with appropriate pre- and post-processing, analog signals are transmitted and certain nomographic functions of the transmitted signals (here, the average) are obtained directly over the air~\cite{Goldenbaum_2013_harnessingInterferenceOTA}. 
\Ac{AirFEEL} enables model aggregation across massive numbers of \acp{ED}  in each \ac{CR} 
without increasing the number of communication resource blocks as the number of participating EDs grows.
Optimal power control for such schemes has been studied in~\cite{Cao_2022_OptimalPowerControlOTAFedAvg,Zhang_Tao_2021_GradientStatPowerContFLOTA}, and for multi-antenna base stations, joint client selection and receive beamforming were designed in~\cite{yang_FL_OTA_2020}. The convergence of such schemes under heterogeneous data statistical distributions has been established in~\cite{Cao_2022_OptimalPowerControlOTAFL, Sery_OTAFLHetero_2021}. 	
Another line of work, reviewed in~\cite{Azimi_OTAFL_SigProcessingReview_2026}, avoids transmitter-side \ac{CSI} by letting EDs transmit at maximum power and extracting information at the receiver via massive-array beamforming~\cite{Amiri_BlindFEEL_2021, Becirovic_MIMO_BlindFEELSparse_2022, Tegin_BlindFEEL_ADCDAC_2021}, weighted aggregation~\cite{Azimi_OTA_FL_WeightedAggregation_2024}, or nested-lattice joint source-channel coding~\cite{Azimi_AirCPU_2026, Razavikia_ChannelComp_2024}. These approaches are attractive when \ac{CSI} at the transmitter is unavailable, but require either large receive arrays or modified aggregation. In contrast, we consider the complementary regime with global and possibly imperfect \ac{CSI} and single-antenna receiver, as commonly assumed in~\cite{yang_FL_OTA_2020, Cao_2022_OptimalPowerControlOTAFedAvg, Cao_2022_OptimalPowerControlOTAFL}.

One of the most promising 6G technologies is \ac{ISAC}~\cite{Liu_ISAC_2022,6Gisac}, which integrates radar sensing and communication functionalities into a unified system. 
Since \ac{AirComp} inherently integrates communication and computation, combining \ac{ISAC} with \ac{AirComp} leads to signals that simultaneously serve sensing, computing, and communication, referred to as \ac{Sig-ISCC}.
Such a framework was first proposed in~\cite{Qi_ISCC_2022} using distinct coded signals for the three functionalities that are superposed at the transmitter. Specifically, sensing, computing (i.e., AirComp), and communication are each served by a dedicated signal with its own precoding vector, and MIMO beamforming techniques are designed to jointly coordinate these signals. A similar framework for vehicular networks was addressed via deep reinforcement learning in~\cite{Yang_SigISCC_SigAddition_DRL_2024} for fast-varying channel environments. 
The fully integrated case (i.e., using the same coded signal for all three functionalities, i.e., joint sensing and AirComp, by considering AirComp as integrated computing and communication) was studied in~\cite{isacota_Li_2023} and subsequently extended to Orthogonal Frequency-Division Multiplexing (OFDM) systems in~\cite{Dong_SigISCC_OFDM_2025}. This framework has also been applied to robotic scenarios, where the sensing performance is characterized by the beampattern difference from the optimal sensing beampattern~\cite{Dong_SigISCC_Robotic_2024}.

We note that, in addition to the Sig-\ac{ISCC} framework that this work considers, two other types of \ac{ISCC} also appear widely in the literature. 
One is the \emph{coexistence design of \ac{ISAC} and computing services}, where ISAC and computation coexist in the network but not through the same signal. The computation may be performed either in a traditional manner via a \ac{MEC} server or via \ac{AirComp} signal superposition. Examples include the coexistence of conventional \ac{ISAC} services with a task-offloading stream supporting \ac{MEC}, as well as the coexistence of downlink \ac{ISAC} with uplink \ac{AirComp}.
Zhao et al.~\cite{Zhao_ISCC_ISACMEC_MultiFunctionalBeamforming_2025} proposed a joint communication and sensing beamforming design for \ac{ISAC}-aided \ac{MEC} systems, in which the \ac{BS} allocates resources across all functionalities and can partially offload computation tasks to nearby edge servers. Resource allocation for a similar system supporting federated learning was studied in~\cite{Liu_ISCC_ISACMEC_MultiTaskLearning_2024} based on multi-task learning. In the context of AirFEEL, the coexistence design has focused on the coexistence of the BS performing ISAC while \acp{ED} conducting uplink aggregation via AirComp under interference from ISAC signals.
Asaad et al.~\cite{Asaad_SigISCC_FEEL_ISAC_SchedulingBeamforming_2025} considered a framework where the \ac{BS} performs downlink \ac{ISAC} during model broadcasting, and the uplink AirComp aggregation may be affected by interference from target echoes of the \ac{BS}'s \ac{ISAC} signals. They proposed joint client scheduling and \ac{BS} beamforming to maximize the number of scheduled \acp{ED} while guaranteeing sensing and AirComp performance. A similar framework was studied in~\cite{Du_SigISCC_SigAddition_FL_2024}, where, after downlink model broadcasting via ISAC, the BS continues downlink sensing while receiving the uplink AirComp model aggregation.

Another type can be termed \emph{functional-level} \ac{ISCC}. It treats sensing as a data acquisition stage for an edge computing or learning task, and communication as the means to transmit sensed data, extracted features, or learned updates to an edge server for computing or fusion.
In this stream, the key coupling lies in how sensing quality, communication distortion and latency, and computing resources jointly affect task-level learning or inference metrics (e.g., the \emph{generalization error}, i.e., the gap between a model's training loss and its expected loss on unseen data).
Recent works study \ac{AirFEEL} while explicitly accounting for sensing noise and aggregation distortions~\cite{Wen_tISCC_Air_FL_2026}, characterizing their joint impact on learning convergence, and \ac{AirComp} has also been adopted for over-the-air fusion of sensed features for edge AI inference~\cite{Zhuang_tISCC_Air_AIInference_2024}, avoiding separate per-device feature transmissions; further extensions along orthogonal directions, such as online learning, unmanned aerial vehicle (UAV)-enabled collection, and generalization-error analyses, can be found in~\cite{Hu_tISCC_Air_OnlineFL_cache_2025,Fu_tISCC_Air_MultiUAV_FL_2025,Liu_tISCC_Air_SparseSpatialFeaturesEdgeAI_2025,Chen_tISCC_ViewAndChannel_EdgeAI_2024,Cai_tISCC_Air_CharacterizeGeneralizationError_2025,Liang_tISCC_normalComm_FL_ResourceAlloccation_2025}.
In parallel, functional-level \ac{ISCC} under conventional digital links has been studied for task-oriented edge inference~\cite{Wen_tISCC_normalComm_TaskOrientedEdgeAI_2023}, in which the sensing and communication resources are co-designed for a downstream inference task rather than for signal reconstruction, and federated edge learning~\cite{Liu_tISCC_normalComm_FL_TaskOriented_2023}, with numerous extensions covering privacy, vertical \ac{FL}, vehicular perception (cooperative environment sensing among vehicles by fusing multi-view on-board sensor data), and aerial deployments (see, e.g.,~\cite{Zhou_tISCC_normalComm_FL_2026,Zhu_SemanticReliability_ISCC_2023,Hu_tISCC_normalComm_DPrivate_WFL_2025,Jiao_tISCC_normalComm_FL_modeldropout_2025,Liu_tISCC_normalComm_VerticalFL_2022,Li_tISCC_normalComm_cooperativePerceptionVehicles_2025,Yang_tISCC_normalComm_HFL_SpaceAirGround_2025}).

This paper focuses on \ac{Sig-ISCC}, which represents a tight form of integration among ISCC integration frameworks, as sensing, computing, and communication are all realized through the same waveform. Such a high degree of integration offers the substantial potential for resource efficiency.
Applying \ac{Sig-ISCC} to \ac{AirFEEL} implies that the \emph{uplink} AirComp signal is also used for the sensing task. The sensing requirement couples non-trivially with the AirComp aggregation quality, necessitating joint resource design. Only a few works in the literature  consider such \ac{Sig-ISCC} for FL.
Du et al.~\cite{Du_SigISCC_FL_preSensingPhase_2025} employed superposed coded sensing and communication signals (as in~\cite{Qi_ISCC_2022}) and exploited the communication component for sensing instead of treating it purely as interference. Pang and Zhang~\cite{Pang_ZhangXi_SigISCC_FL_coopInt_2025} proposed a holistic framework combining signal-level and functional-level \ac{ISCC}, where deep learning modules infer \ac{ED} locations/CSI from received superposed communication signals and then use the inferred CSI for subsequent communication and learning designs. In our previous work~\cite{zheng_FL-OTA-ISAC_2023}, we considered an \emph{individual} sensing requirement for each uplink transmitter, which translates into per-\ac{ED} minimum power constraints; this implies additional sensing transmissions from \acp{ED} that are not selected for AirComp-FL and thus introduces extra interference to AirComp aggregation (see also~\cite[Section 3.3]{Zhu_FL-ISAC-survey_2023}).
However, in many practical deployments, sensing targets and EDs are widely distributed and may exhibit heterogeneous sensing geometries.
Enforcing an individual sensing requirement for each \ac{ED} is unrealistic and unnecessary, as \acp{ED} can jointly perform sensing tasks.
We therefore study such Sig-ISCC for FL with a \emph{joint target detection} task~\cite{cheng_XuJie_ISACDetection_2022} and impose a joint sensing requirement at the \ac{BS} by fusing sensed information across \acp{ED}.

This work therefore investigates Sig-ISCC, with joint detection as the sensing task.
In the considered system, \acp{ED} participate in AirComp aggregation while performing joint detection of a target. 
Conventional AirComp power-control designs typically rely on ranking \acp{ED} according to their channel gains, with the optimal power taking either the maximum power or the channel-inverting level~\cite{Cao_2022_OptimalPowerControlOTAFedAvg,Liu_OTA_TWC_2020}. The additional joint sensing requirement breaks this monotonicity structure due to the coupling of communication and sensing performance, and makes the power allocation substantially more challenging.
The main contributions are summarized as follows:
\begin{itemize}[leftmargin=11pt]
    \item We propose, to the best of our knowledge, the first uplink shared-waveform power control framework for signal-level integrated sensing, computing, and communication (Sig-ISCC) in AirFEEL. Unlike existing works on joint sensing and AirComp-based FEEL, which either operate at the \emph{functional level} by treating sensing as a data collection procedure for FEEL, or consider the \emph{coexistence level} of integration through the coexistence of downlink ISAC and uplink AirComp, our framework exploits a single shared waveform to simultaneously support communication and computation via AirComp-FL, and joint target detection.
    \item We formulate the resulting design as a joint transmit-power and receive-scaling optimization problem, with the objective of minimizing the convergence gap of AirComp-FL while satisfying a joint target-detection requirement. The formulation accounts for imperfect \ac{CSI} by explicitly incorporating the channel estimation error variance, which introduces a power-dependent misalignment term into the aggregation distortion.
    \item Based on our analytical findings, we reformulate the original nonconvex problem as an equivalent convex optimization problem via a proposed variable transformation.
     Building on this, we establish key structural optimality properties showing that the optimal solution can be characterized through the root of a monotone differentiable function. Based on this characterization, we develop a polynomial-time algorithm with provable global optimality and robustness to root-finding inexactness.
    \item Simulation results validate the optimality, confirm the robustness of the proposed algorithm, and demonstrate the importance of optimal power control for AirComp-FL performance under sensing constraints.
\end{itemize}

The remaining sections are organized as follows. Section~II describes the system model. Section~III derives closed-form optimality properties and develops the globally optimal algorithm. Section~IV presents simulation results, and Section~V concludes the paper.

\section{System Model}
\subsection{Federated Learning}

Consider a system with a set $\setK$ of $K$ \acp{ED} and a single-antenna \ac{BS}.
We consider a supervised learning task for FL. \Ac{ED}~$k\in \setK$ serves as a participating client and possesses the local dataset $\setD_k =\{(\vec{u}_{ki},{v_{ki}})\}_{i=1}^{|\setD_k|}$ with $|\setD_k|$ labeled data samples, where $(\vec{u}_{ki},{v_{ki}})$ denotes the~$i$-th data pair of \ac{ED}~$k$, consisting of input data~$\vec{u}_{ki}$ and its label~${v_{ki}}$. 
The local empirical loss function for \ac{ED} $k$ is expressed as~$F_k(\bw) = \frac{1}{|\setD_k|} \sum_{(\vec{u}_{ki},v_{ki}) \in \setD_k} \ell (\bw;\vec{u}_{ki}, v_{ki})$, 
where $\ell(\bw;\vec{u}_{ki}, v_{ki})$ is the loss of the prediction on the sample pair $(\vec{u}_{ki}, v_{ki})$ calculated with model parameters $\bw\in\R^m$. Balanced dataset sizes are assumed throughout the rest of this work as in~\cite{yang_FL_OTA_2020, wang_FL_OTA_RIS_2022}, i.e.,
$|\setD_k| = D_{\mathrm{loc}}$, $\forall k \in \setK$.
The objective of the FL training is to minimize the global loss function $ F(\bw) \defeq \frac{1}{K} \sum_{k \in \setK} F_k(\bw),$
that is, $\min_{\bw\in\R^m} F(\bw)$.
Each \ac{CR} $t$ of FedAvg proceeds as follows: \begin{itemize}[leftmargin=11pt]
    \item The BS broadcasts the current global model ${\bw}^{[t-1]}$ to \acp{ED}.
    \item For each \ac{ED}~$k\in \setK$, 
    one (resp. several~$E$) mini-batch (size $m^{(b)}_k$) stochastic gradient descent step (resp. steps)
    is performed on the local loss~$F_k$ starting from the received global model~${\bw}^{[t-1]}$ using the local dataset~$\setD_k$.
    \item Once all \acp{ED} have completed their local training, they send pilot signals to BS and BS estimates each ED's channel. 
    Based on this information, the power allocation problem (detailed in a later section) is solved and the results are fed back to the \acp{ED}.
Each \ac{ED}~$k$ then sends its gradient (resp. the model difference after several iterations $\bw^{[t]} - \bw_k^{[t+1]}$)~$\bg_k^{[t]}$ to the BS over a communication block of duration~$T$. The channel is assumed to be quasi-static flat fading within each \ac{CR}.
    \item The BS receives the gradient updates and averages them to obtain the aggregated gradient~$\bg^{[t]}$:
    \begin{equation} 
    \bg^{[t]} =\frac{1}{|\setK|} \sum_{k \in \setK}\bg_k^{[t]}.
    \label{eq: average_gradient}
    \end{equation}
    The BS then updates the global model:
    \begin{equation}
        \bw^{[t]} = \bw^{[t-1]} - \gamma^{[t]} \bg^{[t]},
    \end{equation}
    where~$\gamma^{[t]}>0$ is the learning rate at \ac{CR}~$t$
    (resp. $\gamma^{[t]}=1$ if the model-update difference is transmitted).
\end{itemize}

The \ac{CR} index~$t$ is omitted hereafter, as it does not affect the subsequent design.

\subsection{Joint Detection}

Each \ac{ED} is assumed to be equipped with a full-duplex transceiver\footnote{Accordingly, we limit the scope of this work to full-duplex, deep-learning-capable industrial edge nodes~\cite{Luvisotto_full_duplex_Comm_Industrial_2018, Iradier_FullDuplex_IndustrialNetwork_2026} or purpose-built ISAC nodes.} 
with one transmit antenna\footnote{Prior studies have shown that an omnidirectional beampattern can provide effective sensing performance in the high-SNR regime~\cite{Wang_ISCC_BeampatternDesign_2024, Dong_SigISCC_Robotic_2024}. Accordingly, we adopt a single-antenna modeling for optimal power allocation. The formulation can also accommodate directional beamforming by replacing the corresponding channel gain with the effective beamformed gain.}, one receive antenna for communication, and~$N_r$ receive antennas for radar sensing. Coordinated joint target detection is performed following the model in~\cite{Cheng_XuJie_CoordinatedTransmitBeamforming_ISAC_2024}.
The steering vector of the receive antenna array, assumed to be a uniform linear array, is given by $\vec{a}_k(\theta_k)=[1,e^{j2\pi d_a\sin (\theta_k)/\lambda_c},\ldots,e^{j2\pi d_a(N_r-1)\sin (\theta_k)/\lambda_c}]$, where $\lambda_c$ is the carrier wavelength.
If the target is present, each \ac{ED}~$k$ receives the reflected signal as~\cite{Cheng_XuJie_CoordinatedTransmitBeamforming_ISAC_2024}:
\begin{equation}
\vec{r}_k(t)=\sum_{i\in\setK} \bm{G}_{i,k}\sqrt{p_i}s_i[t-\tau_{i,k}]+ \vec{n}_k,
\end{equation}
with $\bm{G}_{i,k} = \sqrt{\beta_{i,k}} \zeta_{i,k}\vec{a}_i(\theta_i)\in\C^{N_r}$  the target response vector from \ac{ED}~$i$ to receiving \ac{ED}~$k$, where $\zeta_{i,k}=\zeta$ denotes the radar cross section (RCS), $\beta_{i,k}$ the round-trip path loss, and $\tau_{i,k}$ the potential echo delay associated with the path from ED~$i$ to ED~$k$ via a target at the hypothesized position~\cite{An_FundamentalDetectionISAC_2023, Chalise_ISAC_UnitaryMatrixDelayDoppler_2017}; $p_i$ is the transmit power of \ac{ED}~$i$;
$s_i$ denotes the waveform containing gradient information in its amplitude for $i\in\setK$, and $\vec{n}_k \sim \mathcal{CN}(0, \sigma_d^2\bm{I}_{N_r})$ is the receiver noise.

The transmitted signal is normalized to have zero-mean and unit variance.
In addition, the transmitted signals of different \acp{ED} are assumed to be independent. Under these assumptions, $\frac{1}{T}\int_{\setT}|s_i(t)|^2 dt=1$ and $\frac{1}{T}\int_{\setT} s_k(t)s_j^*(t-\tau)dt  =0$ for any $\tau>0$ and any $k\neq j$ when $T$ is sufficiently large~\cite{Cheng_XuJie_CoordinatedTransmitBeamforming_ISAC_2024}.
We focus on target-presence detection within a predefined sensing region, rather than general target search or localization. Before operation, a calibration phase with a reference target placed in this region is performed to estimate the RCS~$\zeta$ and the constants governing the round-trip path loss~$\beta_{k,k}$ for all~$k$. Since the sensing-region geometry and the \acp{ED} locations are known, each \ac{ED} can also determine the corresponding potential echo delay.
Consequently, the $k$-th \ac{ED} knows the potential echo delay~$\tau_{k,k}$ of its own transmitted signal (if the target exists), and each \ac{ED} performs matched filtering with its own transmitted signal to extract the reflected signal:
\begin{equation}
	\begin{aligned}
		\vec{d}_{k,k} &= \frac{1}{T}\int_{\setT}\vec{r}_k(t)s_k^*(t-\tau_{k,k}) dt \\
&= \sum_{i\in\setK} \bm{G}_{i,k}\sqrt{p_i}\underbrace{\frac{1}{T}\int_{\setT}s_i(t-\tau_{i,k}) s_k^*(t-\tau_{k,k}) dt}_{\Indic_{(k=i)}} 
		\\
		& \hspace{1cm}+ \underbrace{\frac{1}{T}\int_{\setT}\vec{n}_k(t)s_k^*(t-\tau_{k,k})dt}_{\tilde{\vec{z}}_{k,k}}
	\end{aligned}
\end{equation}
where~$\vec{\tilde{z}}_{k,k}\in\mathcal{CN}(0,\frac{\sigma_d^2}{T^2}\bm{I})$   
is the equivalent noise after the filtering.
All extracted signals $\vec{d}_{k,k}\in\C^{N_r}$, $k\in\setK$, are then aggregated at the BS for joint detection.\footnote{Note that this incurs an additional transmission overhead of $2N_r\cdot32\cdot K$, which is negligible compared with the size of the \ac{AirComp}-FL gradient or model-update transmission.}
 We denote $\bm{G}_{\mathcal{I}} = (\sqrt{p_1}\bm{G}_{1,1}, \sqrt{p_2}\bm{G}_{2,2}, \ldots, \sqrt{p_K}\bm{G}_{K,K})\in\C^{KN_r}$.
The BS determines the presence of the target based on the following binary hypothesis test
\begin{equation}
\begin{cases}
\setH_1: &  \vec{d}_{\mathcal{I}}=\bm{G}_{\mathcal{I}}+\vec{\tilde{z}}_{\mathcal{I}}\\
\setH_0: & \vec{d}_{\mathcal{I}}=\vec{\tilde{z}}_{\mathcal{I}},
\end{cases}
\label{eq: binary hypothesis}
\end{equation}
where the subscript index vector is $\mathcal{I}=(k,k)_{k\in\setK}$.
Under $\setH_1$ and $\setH_0$, the likelihoods are complex Gaussian with means $\bm{G}_{\mathcal{I}}$ and $\bm{0}$, respectively, and common covariance $\frac{\sigma_d^2}{T^2}\bm{I}$.
A likelihood ratio test is performed, $\Lambda (\vec{d}_{\mathcal{I}}) = \frac{ p(\vec{d}_{\mathcal{I}} | \setH_1) }{ p(\vec{d}_{\mathcal{I}} | \setH_0) } \mathop{\gtrless}_{\mathcal{H}_0}^{\mathcal{H}_1} e^{\delta},$ with some $\delta\in\R$. Taking the logarithm yields
\begin{equation}
	\Re\{\vec{d}_{\mathcal{I}}^H\bm{G}_{\mathcal{I}}\}  \mathop{\gtrless}_{\mathcal{H}_0}^{\mathcal{H}_1} \delta'
	\label{eq: final ratio test}
\end{equation}
where $\delta' \defeq \frac{1}{2}\big(\frac{\sigma_d^2}{T^2}\delta + \|\bm{G}_{\mathcal{I}}\|^2\big)$.
The term $\vec{d}_{\mathcal{I}}^H\bm{G}_{\mathcal{I}}$ follows a circularly symmetric complex Gaussian distribution with
\begin{equation}
	\begin{cases}
		\! \setH_1\!: & \hspace{-.2cm} \EX[\vec{d}_{\mathcal{I}}^H\bm{G}_{\mathcal{I}}] \!=\! \|\bm{G}_{\mathcal{I}}\|^2,\  \Var[\vec{d}_{\mathcal{I}}^H\bm{G}_{\mathcal{I}}] \! =\!  \Var[ \vec{\tilde{z}}_{\mathcal{I}}^H \bm{G}_{\mathcal{I}}]  \\
	\!	\setH_0\! : & \hspace{-.2cm} \EX[\vec{d}_{\mathcal{I}}^H\bm{G}_{\mathcal{I}}] \!=\!  0,\  \Var[\vec{d}_{\mathcal{I}}^H\bm{G}_{\mathcal{I}}] \!=\! \Var[ \vec{\tilde{z}}_{\mathcal{I}}^H \bm{G}_{\mathcal{I}}]. 
	\end{cases}
\end{equation}
Moreover, $ \Var[ \vec{\tilde{z}}_{\mathcal{I}}^H \bm{G}_{\mathcal{I}}]= \frac{\sigma_d^2}{T^2} \|\bm{G}_{\mathcal{I}}\|^2$.
We denote $\setE_{\mathcal{I}}\defeq  \|\bm{G}_{\mathcal{I}}\|^2 = \sum_{k\in\setK} \zeta^2N_r  \beta_{k,k}p_k =  \sum_{k\in\setK}b_kp_k $, where we define $b_k \defeq\zeta^2N_r\beta_{k,k} $.
Therefore,
\begin{equation}
	\begin{cases}
		\setH_1:& {\Re\{\vec{d}_{\mathcal{I}}^H \bm{G}_{\mathcal{I}}\}} \sim \mathcal{N}\Big(\setE_{\mathcal{I}}, \frac{\sigma_d^2\setE_{\mathcal{I}}}{2T^2}\Big) \\
		\setH_0:&  {\Re\{ \vec{d}_{\mathcal{I}}^H \bm{G}_{\mathcal{I}} \}} \sim \mathcal{N}\Big(0, \frac{\sigma_d^2\setE_{\mathcal{I}}}{2T^2}\Big).
	\end{cases}
\end{equation}
Consequently, the detection probability 
\begin{equation}
	p_{\mathrm{D}}\! =\! \Prob\! \Big(\! \Re\{\vec{d}_{\mathcal{I}}^H \bm{G}_{\mathcal{I}}\} \! \geq \!\delta' \mid \setH_1 \Big)\! =\! Q\Bigg((\delta' \!-\! \setE_{\mathcal{I}})\sqrt{\frac{2T^2}{\sigma_d^2\setE_{\mathcal{I}} }}\Bigg),
\end{equation}
with $Q$ denoting the Gaussian Q-function. The false alarm probability can be written as:
\begin{equation}
	p_{\mathrm{FA}} = \Prob \Big(\Re\{\vec{d}_{\mathcal{I}}^H \bm{G}_{\mathcal{I}}\} \geq \delta' \mid \setH_0 \Big) = Q\Bigg(\delta' \sqrt{\frac{2T^2}{\sigma_d^2\setE_{\mathcal{I}} }}\Bigg).
\end{equation}
Therefore, 
with a log-likelihood ratio test, for a sufficiently long transmission time~$T$, the Neyman-Pearson detector gives a detection probability under a given false alarm probability $p_{\mathrm{FA}}$ as:
\begin{equation}
p_{\mathrm{D}} = Q\Bigg(Q^{-1}(p_{\mathrm{FA}}) - \sqrt{\frac{2T^2\setE_{\mathcal{I}}}{\sigma_d^2}}\Bigg),
\end{equation}
 
The detection probability $p_{\mathrm{D}}$ is monotonically increasing with respect to $\setE_{\mathcal{I}}$. Therefore, given a false alarm probability and a detection probability threshold $p_{\mathrm{D,th}}$, the minimum detection probability requirement can be expressed as
\begin{equation}
p_{\mathrm{D}} \geq p_{\mathrm{D,th}}\  \Leftrightarrow \ \setE_{\mathcal{I}} \geq \eta_{\mathrm{D}} \defeq p_{\mathrm{D}}^{-1}(p_{\mathrm{D,th}}),
\end{equation}
where $p_{\mathrm{D}}^{-1}$ is the inverse function of $p_{\mathrm{D}}:\setE_{\mathcal{I}}\mapsto p_{\mathrm{D}}(\setE_{\mathcal{I}})$. This translates into the following sensing constraint:
\begin{equation}
     \sum_{k\in\setK} p_kb_k \geq \eta_{\mathrm{D}}.
\label{eq: sensing_cstr}
\end{equation}

\begin{remark}[Extensibility]\label{remark: sensing extension}
	With known hypothesis priors and error costs, the sensing model can be extended to Bayesian detection. Multi-hypothesis testing, e.g., for multiple targets, and clutter-interference extensions are left for future work.
\end{remark}

\subsection{Over-the-Air Computation (\ac{AirComp})}\label{subsec: AirComp}

All selected \acp{ED} transmit their gradients over the same resource block via analog transmission. Perfect synchronization is assumed~\cite{Chen_2024_OTA_OFDM}.\footnote{In practice, coarse synchronization combined with timing advance compensate propagation delays~\cite{Chen_2024_OTA_OFDM}, with residual timing offsets accommodated by a cyclic prefix or guard interval and the associated phase offset incorporated into the estimated effective channel.
	}
The BS receives:
\begin{equation}
    \by = \sum_{k\in\setK} h_k \sqrt{p_k} \overline{\bg}_k +\bn,
\end{equation}
where~$h_k= \hat{h}_k + \varepsilon_{h, k}$ is the actual channel coefficient, where $\hat{h}_k>0$ denotes the estimated CSI and $\varepsilon_{h, k} \sim \mathcal{CN}(0, \sigma_{h,k}^2)$ denotes the estimation error, modeled as a complex Gaussian random variable~\cite{Chen_2024_OTA_OFDM, Xie_OTA_CSI_2024}\footnote{Equivalently, with complex channel $g_k$ and estimate $\hat g_k$, transmitter-side compensation by $\hat g_k^*/|\hat g_k|$ gives
		$
		g_k\!\frac{\hat g_k^*}{|\hat g_k|}
		\!=\! |\hat g_k|\!+\!\varepsilon_{h,k}.
		$
		Thus, we set $\hat h_k\!=\!|\hat g_k|\!>\!0$ and include the phase-compensation error in $\varepsilon_{h,k}$.
};
$p_k$ is the transmit power; $\bn\sim \mathcal{CN}(0, \sigma_n^2\bm{I}_m)$ is Gaussian noise.
It is assumed that the target is sufficiently far from the \ac{BS} so that the interference from echo signals reflected off the target is negligible compared to the \ac{AirComp} communication signals\footnote{Due to two-way path loss, simulations show that the target-reflected echo contributes less than $0.1\%$ to \MSEOTA; otherwise, it can be modeled as an additional distortion term that impacts the MSE the same way as CSI error.}.
The participating gradients are normalized using a common pooled mean and standard deviation as in~\cite{zhu_OTA-FL_2020,zheng_FL-OTA-ISAC_2023}:\footnote{Zero-mean and unit-variance normalization improves signal recovery under the per-device average power constraint by concentrating transmit power on informative signal variations.}
$
    \overline{\bg}_k = \frac{\bg_k - \mu}{\Gamma},
$
where the normalizing mean~$\mu=\frac{1}{|\setK|} \sum_{k\in\setK}\mu_k$ and pooled standard deviation~$\Gamma=\sqrt{\frac{1}{|\setK|}\sum_{k\in\setK}(\Gamma_k^2+\mu_k^2)-\mu^2}$ are obtained from the statistics of participating \acp{ED}:
$
    \mu_k = \frac{1}{m}\sum_{i=1}^m g_{k,i}$, $\Gamma_k = \sqrt{\frac{1}{m}\sum_{i=1}^m g_{k,i}^2 - \mu_k^2},
$
where $\bg_k=(g_{k,i})_{i=1,\ldots,m}$ is the gradient vector for \ac{ED}~$k$.
Note that $\mu$ and $\Gamma$ must be uniform across all \acp{ED} due to the nature of AirComp and are estimated prior to data transmission. After this normalization, and assuming that differently shifted versions of $\bar{g}_k$ and $\bar{g}_j$ are uncorrelated for $k\neq j$, the signal conditions required for joint detection are satisfied.
The maximum transmit power constraint~$P_{\max}$ for each \ac{ED}~$k$ transmitting at power $p_k$ is given by:
\begin{equation}
\frac{1}{m} \EX \|\sqrt{p_k} \overline{\bg}_k\|^2=p_k\leq P_{\max}.
\end{equation}
Upon receiving the superimposed AirComp signals through the multi-access channel, the BS rescales the received signal, reverses the normalization, and uses the resulting expression~$\tilde{\bg}$ as the estimated gradient for~$\bg$ in~\eqref{eq: average_gradient}:
$\tilde{\bg} = \frac{\Gamma\sqrt{\alpha}\by}{|\setK| } + \mu$,
with~$\alpha> 0$ the receive scaling factor.

The gradient recovery error can be derived as:
\begin{equation}
\begin{aligned}
    \varepsilon & = \estimG - \bg\\
    &=  \frac{\Gamma\sqrt{\alpha}\by}{|\setK| } + \mu - \frac{1}{|\setK|} \sum_{k \in \setK}(\Gamma \overline{\bg}_k+ \mu)\\
    &= \frac{\Gamma}{|\setK|} \Bigg(
    \sum_{k\in\setK} \Big((\hat{h}_k + \varepsilon_{h,k}) \sqrt{\alpha p_k} - 1\Big) \overline{\bg}_k + \sqrt{\alpha}\bn
    \Bigg).
\end{aligned}
\end{equation}

The resulting \ac{MSE} of AirComp accounting for channel estimation error is therefore~\cite{Cao_2022_OptimalPowerControlOTAFL}, \cite{Chen_2024_OTA_OFDM, Xie_OTA_CSI_2024}:
\begin{equation}
\begin{aligned}
    &\MSEOTA (\bm{p}, \alpha) = \EX \|\estimG - \bg\|^2 \\
  \!\! \approx &  \Gamma^2m \underbrace{\frac{1}{K^2}\!\Big[\! \sum_{k\in\setK}(\hat{h}_k\sqrt{p_k\alpha} \!-\! 1)^2  \!+\!  \alpha\! \sum_{k\in\setK} p_k\sigma_{h,k}^2 \!+\! \alpha\sigma_n^2
   \Big]}_{\defeq \tMSEOTA},
\end{aligned}
\label{eq: MSE_OTA}
\end{equation}
where $\bm{p}=(p_k)_{k\in\setK}$ denotes the power allocation vector. Compared with the perfect-CSI case, the channel estimation error contributes the additional term $\alpha\sum_{k\in\setK}p_k\sigma_{h,k}^2$, which acts as a self-interference term that grows with the transmit power and the receive scaling factor~$\alpha$; \tMSEOTA\,is the scale- and  model parameter-normalized AirComp MSE.
\subsection{Sig-ISCC Protocol}
The Sig-ISCC communication protocol is summarized in \Cref{fig: illustration_system}. Assuming that downlink model broadcasting and local training have been completed, the protocol proceeds as follows:
\begin{itemize}[leftmargin=13pt]
	\item \emph{Sensing Calibration:} Calibration estimates the sensing parameters, including RCS, round-trip path loss, and static clutter, equivalently $\bm{G}_{k,k}$. These parameters are reported to the \ac{BS} to derive $b_k$.
	
	\item \emph{Normalization Factor Upload:} Each \ac{ED}~$k$ uploads $\mu_k$ and $\Gamma_k$ over a dedicated control signaling resource, from which the \ac{BS} obtains the uniform normalization factors $\mu$ and $\Gamma$.
	
	\item \emph{Uplink Pilot and Resource Allocation:} Each \ac{ED} transmits an uplink pilot~\cite{Xie_OTA_CSI_2024}. The \ac{BS} estimates $\hat{h}_k$ with previously estimated error variance $\sigma_{h,k}^2$~\cite{Chen_2024_OTA_OFDM, Xie_OTA_CSI_2024}, and solves the joint power-allocation and receive-scaling problem for $(\lambda^*,\alpha^*)$ in Section~III-C.
	
	\item \emph{Normalization and Power Feedback:} The \ac{BS} broadcasts $\mu,\Gamma$ and $(\lambda^*,\alpha^*)$, from which each \ac{ED} derives $p_k$ using \Cref{theorem: closed form solution of x_k}. If the \acp{ED} do not already know $\hat{h}_k$ (for example from time-division duplexing), the \ac{BS} also feeds back the CSI~\cite{Xie_OTA_CSI_2024}.
	
	\item \emph{Sig-ISCC Transmission:} Each \ac{ED} normalizes its model update using $\mu$ and $\Gamma$ and transmits $\overline{\bg}_k$ with power $p_k$.
	
	\item \emph{Sensing Result Upload:} Each \ac{ED} uploads its matched-filter output $\vec{d}_{k,k}$, and the \ac{BS} performs the log-likelihood-ratio test in~\eqref{eq: final ratio test}.
\end{itemize}

\begin{figure}[t]
	\centering
\includegraphics[width=0.85\linewidth]{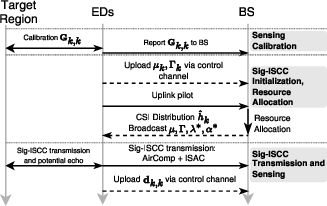}
	\caption{System communication protocol.}
	\label{fig: illustration_system}
\end{figure}

\subsection{Convergence Analysis of This AirComp-FL Framework}
 The convergence of AirComp-based FedAvg has been established in~\cite[Theorem~1]{Cao_2022_OptimalPowerControlOTAFedAvg}. The proof accounts for the local model drift w.r.t. $E\geq1$ local epochs on non-i.i.d. local data, quantifying the local model drift.

\begin{assumption}[Smoothness]
\label{assumption: smoothness}
The loss functions $F_k$, $k\in\setK$, are differentiable and for any $\bw_1,\bw_2\in\R^m$, $\|\nabla F_k(\bw_1)-\nabla F_k(\bw_2)\|\leq L\|\bw_1-\bw_2\|$.
\end{assumption}
\begin{assumption}[Polyak-\L{}ojasiewicz Inequality]
\label{assumption: PL inequality}
There exists $\mu>0$ such that, for any $\bw\in\R^m$, $\|\nabla F(\bw)\|^2\geq2\mu(F(\bw)-F^*)$, with $F^*=\min_{\bw}F(\bw)$.
\end{assumption}
The PL inequality is a standard assumption in the \ac{SGD} convergence analysis~\cite{signSGD_proofConvergence_2018,Karimi_PLcondition_2016} and in AirComp-FL~\cite{Cao_2022_OptimalPowerControlOTAFL,Cao_2022_OptimalPowerControlOTAFedAvg}. It is strictly weaker than strong convexity, as it requires neither convexity nor a unique minimizer.
\begin{assumption}[Variance Bound]
\label{assumption: variance bound}
Let~$\xi_k$ denote a mini-batch sampled uniformly at random from~$\setD_k$. For every \ac{ED}~$k\in\setK$, the local stochastic gradient satisfies $\EX_{\xi_k}[\nabla F_k(\bw;\xi_k)] = \nabla F_k(\bw)$ and $\EX_{\xi_k}\big[\|\nabla F_k(\bw;\xi_k) - \nabla F_k(\bw)\|^2\big] \leq \frac{\varsigma_k^2}{m_k^{(b)}}$.
\end{assumption}
\begin{assumption}[Gradient Bound]
\label{assumption: bounded moment}
For every \ac{ED}~$k\in\setK$, $\EX_{\xi_k}\big[\|\nabla F_k(\bw;\xi_k)\|^2\big] \leq G_k^2$.
\end{assumption}
\begin{assumption}[Gradient Divergence]
\label{assumption: gradient divergence}
There exist~$\chi_k>0$, $k\in\setK$, such that $\|\nabla F_k(\bw)-\nabla F(\bw)\|\leq\chi_k$ for all $\bw\in\R^m$.
\end{assumption}
Assumption~\ref{assumption: gradient divergence} quantifies the non-i.i.d. degree across \acp{ED} directly as a bound on how far each ED's local gradient can drift from the global gradient.

\begin{theorem}[Derived from~\cite{Cao_2022_OptimalPowerControlOTAFedAvg}]
\label{th: AirFEEL convergence analysis}
Let Assumptions \ref{assumption: smoothness}-\ref{assumption: gradient divergence} hold. Let $B\defeq\frac1K\sum_{k\in\setK}\frac{\chi_k^2+\varsigma_k^2}{2}$ and $V\defeq\frac LK\sum_{k\in\setK}G_k^2$. 
Suppose diminishing learning rates are applied $\gamma^{[t]}=\frac{\beta}{t+a}$, $\forall t\in\{1,\ldots,T\}$, with $a>0$ and $\beta\geq\frac{1}{\mu(E-1)}$, such that $\gamma^{[1]}\leq\frac{1}{LE}$. Define $C_t\defeq1-(E-1)\mu\gamma^{[t]}$ and $J_t\defeq\prod_{\tau=1}^TC_\tau/C_t$. The expected optimality gap of Air-FedAvg satisfies
\begin{multline}
    \EX\big[F(\bw^{[T]})\big]-F^*\leq\prod_{t=1}^TC_t\big(F(\bw^{[0]})-F^*\big)\\
    +\sum_{t=1}^TJ_t\Big(\gamma^{[t-1]}EB+(\gamma^{[t-1]})^2E^2V\Big)\\
   \! +\! \underbrace{\sum_{t=1}^T\! \frac{J_t}{2}\Big(\frac{1}{\gamma^{[t-1]}}\!+\!L^2\! \gamma^{[t\! -\!1]}E\! +\! L\Big)\MSEOTA^{[t]}(\bm{p}^{[t]},\alpha^{[t]})}_{\Delta\text{: AirComp-dependent optimality gap}},
    \end{multline}
with $\MSEOTA^{[t]}(\bm{p}^{[t]},\alpha^{[t]})$ the MSE
at the communication round~$t$.
\end{theorem}
\begin{proof}
Recall that $\varepsilon^{[t]}\defeq\tilde{\bg}^{[t]}-\bg^{[t]}$ denote the aggregation error of \ac{CR}~$t$ (Section~II.B), \cite[Theorem~1]{Cao_2022_OptimalPowerControlOTAFedAvg} shows that, under Assumptions~\ref{assumption: smoothness}-\ref{assumption: gradient divergence} and the learning-rate schedule above, the expected optimality gap of Air-FedAvg satisfies
\begin{multline}
    \EX\big[F(\bw^{[T]})\big]-F^*\leq\prod_{t=1}^TC_t\big(F(\bw^{[0]})-F^*\big)\\
    +\sum_{t=1}^TJ_t\Big(\gamma^{[t-1]}EB+(\gamma^{[t-1]})^2E^2V\Big)\\
  \hspace{-.3cm}  +\sum_{t=1}^T\frac{J_t}{2}\Big(\frac{1}{\gamma^{[t\!-\! 1]}}\big\|\! \EX[\varepsilon^{[t]}] \big\|^2 \!+\! \big(L^2\! \gamma^{[t\!-\! 1]}E\!+\! L\big)\! \EX\! \big[\|\varepsilon^{[t]}\|^2\big]\! \Big).
    \label{eq: Cao theorem 1 original}
\end{multline}
Via Jensen's inequality (due to the convexity of $\|\cdot\|^2$),
\begin{equation}
\big\|\EX[\varepsilon^{[t]}]\big\|^2 \leq \EX\big[\|\varepsilon^{[t]}\|^2\big].
\end{equation}
By definition (Section~II.B), $\EX\big[\|\varepsilon^{[t]}\|^2\big]=\MSEOTA^{[t]}(\bm{p}^{[t]},\alpha^{[t]})$. Substituting both into~\eqref{eq: Cao theorem 1 original} yields the claimed bound.
\end{proof}
\begin{remark}
Cao et al.~\cite{Cao_2022_OptimalPowerControlOTAFedAvg} optimize a weighted sum of the signal misalignment and MSE terms in~\eqref{eq: Cao theorem 1 original}, with weights determined by $L$, $\mu$, and the scheduled learning rate $\gamma^{[t]}$. However, these assumptions may not hold for modern deep learning models~\cite{xu_localPLdescentlemma_2025, zhang_gradientclippingacceleratestraining_2020}, while the corresponding constants are difficult to estimate~\cite{jordan_exactlycomputinglocallipschitz_2021} and can yield overly conservative designs. Moreover, such theoretical learning-rate scheduling in the form of $\gamma^{[t]}=\frac{\beta}{t+a}$~\cite{Bottou_SGD_LargeScale_2018} is rarely used in practice~\cite{defazio_optimallineardecaylearning_2024}. We therefore apply Jensen's inequality to combine the two terms into a single objective independent of $L$, $\mu$, and $\gamma^{[t]}$.
\end{remark}

\section{Problem Formulation and Optimal Power Control Design}

\subsection{Problem Formulation}

All \acp{ED} in $\setK$ participate in both AirComp aggregation and coordinated sensing. The physical-layer design problem is to optimize, at every \ac{CR}, the transmit powers and the receive scaling factor so as to minimize the optimality gap~$\Delta$ in \Cref{th: AirFEEL convergence analysis} while satisfying sensing constraints. {The only term in the optimality gap that depends on the power allocation and receive scaling factor is $\MSEOTA(\bm{p}, \alpha)$}; therefore:
\begin{equation}
    \min_{\bp, \alpha} \Delta \iff \min_{\bp, \alpha} \MSEOTA(\bm{p}, \alpha).
\end{equation}

The problem becomes minimizing the MSE of AirComp aggregation under the sensing constraint~\eqref{eq: sensing_cstr}.
This coupling is non-trivial because communication and sensing favor different system parameters. In particular, an \ac{ED} with a large sensing coefficient~$b_k$ may need a high transmit power to satisfy the joint detection constraint, whereas AirComp prefers aligned received power across the devices in~$\setK$.

Accordingly, we formulate the following joint power-control and receive-scaling problem:
\begin{subequations}
\begin{align}
\mathrm{(P0):}\!\!\!\!\!\!\!\!\!\min_{\{p_k\}_{k\in\setK}, \alpha>0}  &  \sum_{k\in\setK}\!(\hat{h}_k\sqrt{p_k\alpha} \!-\! 1)^2  \!+\!  \alpha \!\sum_{k\in\setK} p_k\sigma_{h, k}^2 \!+\! \alpha\sigma_n^2  \label{pb: init}\\
\mathrm{s.t.} \quad\quad
& \sum_{k\in\setK} p_kb_k \geq \eta_{\mathrm{D}}, \label{cons: sensing} \\
& 0 \leq p_k \leq P_{\max}, \  \forall k \in \setK. \label{cons: power}
\end{align}
\end{subequations}

\subsection{Optimization Problem Solution Existence and Convex Transformation}
\label{section: optimal power control given client selection}

The joint power-control and receive-scaling problem can be written as:
\begin{subequations}
\label{pb: given_client_selection}
\begin{align}
\mathrm{(P1)}\!:\hspace{-.6cm}\min_{ \{p_k\}_{k\in\setK},\alpha>0}  &
\alpha(\sigma_n^2 \! +\! \sum_{k\in\setK}(\hat{h}_k^2 \!+\! \sigma_{h,k}^2) p_k) \! -\! 2\sum_{k\in\setK}\! \hat{h}_k \! \sqrt{\! \alpha p_k} \\
\mathrm{s.t.} \quad\quad & \eqref{cons: sensing}, \eqref{cons: power} \notag
\end{align}
\end{subequations}
For notational convenience, we introduce the \emph{effective channel gain}
\begin{equation}
	\tilde{h}_k^2 \defeq \hat{h}_k^2 + \sigma_{h,k}^2, \quad k\in\setK,
	\label{eq: effective channel gain}
\end{equation}
which is the conditional second moment $\EX[|h_k|^2\mid \hat{h}_k]$ of the true channel given its estimate. The objective of $\mathrm{(P1)}$ is then written compactly as $\alpha(\sigma_n^2 + \sum_{k\in\setK}\tilde{h}_k^2 p_k) - 2\sum_{k\in\setK}\hat{h}_k\sqrt{\alpha p_k}$.
All results derived hereafter can be reduced to the perfect-CSI case by setting $\sigma_{h,k}^2=0$, i.e., $\tilde{h}_k^2=\hat{h}_k^2$.

\begin{proposition}(Feasibility Condition)
\label{theorem: feasibility condition}
    The problem~$\mathrm{(P1)}$ is feasible if and only if 
\begin{equation}
    P_{\max} \sum_{k\in\setK}b_k \geq \eta_{\mathrm{D}},
    \label{eq: feasible condition}
\end{equation}
and strictly feasible if 
\begin{equation}
    P_{\max} \sum_{k\in\setK}b_k > \eta_{\mathrm{D}}.
    \label{eq: strict feasible condition}
\end{equation}
\end{proposition}

Since $\alpha>0$ is defined on a non-closed set, one must first establish the existence of an optimum for this optimization problem.
\begin{theorem}[Existence of minimum of $\mathrm{(P1)}$]
\label{theorem: solution existence}
When~\eqref{eq: feasible condition} holds, the optimum of the problem~$\mathrm{(P1)}$ exists.
\end{theorem}
\begin{proof}
See Appendix~\ref{proof: theorem solution existence}.
\end{proof}

We assume \emph{strict feasibility} throughout the rest of this work, since the optimization under equality in~\eqref{eq: feasible condition} is trivial. In fact, if $P_{\max} \sum_{k\in\setK}b_k=  \eta_{\mathrm{D}}$, then clearly $p_k^* = P_{\max}$, and the optimum $\alpha^*$ can be easily derived by zeroing the first derivative of the objective with respect to $\alpha$.
\begin{assumption}
    \label{assumption: strict feasible}
The strict feasibility condition~\eqref{eq: strict feasible condition} holds, i.e., there exists a strictly feasible point for the problem~$\mathrm{(P1)}$.
\end{assumption}

Problem~$\mathrm{(P1)}$ is non-convex due to the product coupling of $p_k$ and $\alpha$ in the objective.
By introducing a bijective change of variables: 
\begin{equation}
    \begin{cases}
        \R_{\geq 0}^K\times\R_{>0} &\rightarrow \R_{\geq 0}^K\times\R_{>0}\\
            ((p_k)_k, \alpha) &\mapsto ( (p_k\alpha)_k, \alpha) = ((x_k)_k, \alpha).
    \end{cases}
    \label{eq: variable transformation pk xk}
\end{equation}
 $(\mathrm{P1})$ is equivalent to the following optimization problem
 for $k\in\setK$:
\begin{subequations}
\label{pb: var_transformed}
\begin{align}
\mathrm{(P2)}:\hspace{-.4cm}\min_{ \{x_k\}_{k\in\setK},\alpha>0}  &
\sum_{k\in\setK}\tilde{h}_k^2x_k - 2\sum_{k\in\setK}\hat{h}_k \sqrt{x_k} + \alpha\sigma_n^2\\
\mathrm{s.t.} \quad\quad &  \sum_{k\in\setK} b_kx_k \geq \eta_{\mathrm{D}} \alpha  \label{cons (pb xk): sensing} \\ & 0 \leq x_k \leq  P_{\max} \alpha, \  \forall k \in \setK.\label{cons (pb xk): power} 
\end{align}
\end{subequations}

The resulting optimization problem is \emph{convex} since the objective function is convex and the constraints are affine.
Although $\mathrm{(P2)}$ can be solved numerically using standard convex solvers, we further conduct an in-depth theoretical analysis to reveal its optimal solution structure and enable a lightweight solver-independent implementation. Specifically, the analysis will characterize a finite set of active-set and reduce the computation to one-dimensional root finding over analytically characterized intervals. The resulting subproblems are independent and can be evaluated in parallel, enabling predictable and efficient repeated power-control decisions. Some numerical advantages are also shown in Section~IV.
We first establish a strict positivity property of the optimal $x_k^*$ for $k\in\setK$.
\begin{proposition}
\label{lemma: p positive}
Every optimal solution satisfies 
$x_k^*>0$ for all $k\in\setK$.
\end{proposition}

\begin{proof}
In Appendix~\ref{proof: lemma p positive}.
\end{proof}

\begin{lemma}[Special case without sensing]
\label{lemma: res_lambda_zero}
For $\mathrm{(P2)}$,
if the constraint~\eqref{cons (pb xk): sensing} is relaxed, the solution is as follows. Assume the \acp{ED} are indexed such that $\hat{h}_1/\tilde{h}_1^2\leq \cdots \leq \hat{h}_K/\tilde{h}_K^2$. There exists $\ell\in\{1,\ldots, K\}$ such that
$\forall k\leq \ell-1, \ x_k^{(\ell)} = \frac{\hat{h}_k^2}{\tilde{h}_k^4}\leq\alpha^{(\ell)}P_{\max}$, and $\forall k\geq \ell, \ x_k^{(\ell)} = \alpha^{(\ell)} P_{\max} \leq \frac{\hat{h}_k^2}{\tilde{h}_k^4}$ with
\begin{equation}
\alpha^{(\ell)} = \frac{1}{P_{\max}}\Big(\frac{\sum\limits_{k\geq\ell} \hat{h}_k }{\frac{\sigma_n^2}{P_{\max}}+\sum\limits_{k\geq\ell} \tilde{h}_k^2}\Big)^2.
\end{equation}
Such $(x_k^{(\ell)}, \alpha^{(\ell)})$ is feasible and optimal for $\mathrm{(P2)}$ without the sensing constraint. 
\end{lemma}
This optimal solution form (without sensing) is common in the literature, e.g. \cite[Theorem~1]{Zhang_Tao_2021_GradientStatPowerContFLOTA} without considering the cross-product term in the MSE expression and \cite[Theorem~1]{CaoXiaowen_OTA_PowerControl_2020}, with $\sigma_{h,k}^2=0$.

\subsection{Theoretical Analysis of $\mathrm{(P2)}$}

Let $\lambda, \xi_k, \rho_k, \nu\geq 0$ denote the dual variables. The Lagrangian function is written as:
\begin{multline}
\hspace{-.2cm}
L(\{x_k\}_k,\alpha, \lambda,\{\xi_k\}_k, \{\rho_k\}_k,\nu) \!=\! \sum_{k\in\setK}\!\tilde{h}_k^2x_k -2\! \sum_{k\in\setK}\! \hat{h}_k \sqrt{x_k}   \\
+\alpha\sigma_n^2 
+ \lambda(\eta_{\mathrm{D}}\alpha - \sum_{k\in\setK} b_kx_k) - \sum_{k\in\setK}\xi_kx_k  \\
+ \sum_{k\in\setK}\rho_k(x_k - P_{\max}\alpha) -\nu\alpha.
\end{multline}

The KKT conditions can be written as
\begin{subnumcases}{\phantom }
\text{Primal constraints:~\eqref{cons (pb xk): sensing},\eqref{cons (pb xk): power}, }\alpha^*>0  \notag \\
\text{D: } \lambda^*,\xi_k^*,\rho_k^*,\nu^* \geq 0 \notag\\
\text{C.S.: }\lambda^*(\eta_{\mathrm{D}}\alpha^* - \sum_{k\in\setK} b_kx_k^*) =0, \notag\\
\quad\quad\xi_k^*x_k^*=0, \rho_k^*(x_k^* - P_{\max}\alpha^*)=0,
\nu^*\alpha^* = 0 \notag\\
(\forall k\in\setK)\ \nabla_{x_k}\! L\!=\! \tilde{h}_k^2 \!-\! \frac{\hat{h}_k }{\sqrt{x_k^*}} \!-\! \lambda^* b_k \!-\! \xi_k^* \!+\! \rho_k^* \!=\! 0, \label{KKT: grad_xk_Sc}\\
\nabla_{\alpha}L=\sigma_n^2+\lambda^* \eta_{\mathrm{D}} - P_{\max}\sum_{k\in\setK}\rho_k^* - \nu^* = 0.\label{KKT: grad_alpha}
\end{subnumcases}

\begin{proposition}
\label{theorem: closed form solution of x_k}
Under \Cref{assumption: strict feasible}, the optimum of $\mathrm{(P2)}$, i.e., $(x_k^*, \alpha^*)$ with $\lambda^*\geq 0$ (the optimal dual variable of the sensing constraint), satisfies the following conditions:
For $k\in\setK$:
\begin{equation}
x_k^*\! \!=\!x_k^*(\lambda^*\!\!,\alpha^*) \!\defeq\! \Indic_{\! \{\! \lambda^*\!<\!\iota_k\!(\alpha^*\!)\}}\! \frac{\hat{h}_k^2}{(\tilde{h}_k^2\!-\!\lambda^* b_k)^2}
 +\! \Indic_{\! \{\lambda^*\! \geq\! \iota_k(\alpha^*)\}} \alpha^{\! *}\! P_{\max},
 \label{eq: xk Sc expression}
\end{equation}
where $\iota_k(\alpha^*) \defeq \frac{\tilde{h}_k^2}{b_k}\big(1- \frac{\hat{h}_k}{\tilde{h}_k^2\sqrt{\alpha^* P_{\max}}}\big)$; $\Indic_{E}$ is the indicator function of set $E$.

\end{proposition}

\begin{proof}
    The optimization problem is convex, and Slater's condition holds under the strict feasibility condition.
The KKT conditions are therefore necessary and sufficient for optimality. By complementary slackness, we clearly have $\nu^*=0$ and for any $k\in\setK$, $\xi_k^*=0$ by~\Cref{lemma: p positive}.
For~$k\in\setK$, if $0<x_k^*<P_{\max}\alpha^*$, then $\rho_k^*=0$, and we obtain:
$x_k^* = \frac{\hat{h}_k^2 }{(\tilde{h}_k^2 - \lambda^* b_k)^2}$.
We aim to find the exact region where $x_k^*<P_{\max}\alpha^*$,
\begin{equation}
\hspace{-.2cm}\frac{\hat{h}_k^2 }{(\tilde{h}_k^2 \!-\! \lambda^* b_k)^2} \!<\! P_{\!\max} \alpha^*
\!\iff\!  \lambda^* \!\!<\!  \frac{\tilde{h}_k^2}{b_k}\! \Big(1\!-\! \frac{\hat{h}_k}{\tilde{h}_k^2\sqrt{\alpha^*\! P_{\!\max}}}\Big)\!.
\label{eq: equivalence for lambda expression}
\end{equation}
This also implies that if \ac{ED}~$k$ does not transmit at maximum power, then necessarily (since $\lambda^* \geq 0$):
$\alpha^* P_{\max} > \hat{h}_k^2/\tilde{h}_k^4$.
\end{proof}
\Cref{theorem: closed form solution of x_k} enables ED~$k$ to recover $x_k^*$ from its locally known $h_k$ and $b_k$ given the broadcasted $(\lambda^*,\alpha^*)$.

\begin{remark}[Role of the channel estimation error]
\label{remark: impact of CSI error}
For fixed $(\lambda,\alpha)$, increasing $\sigma_{h,k}^2$ decreases the interior $x_k^*$ and increases $\iota_k(\alpha)$, making the maximum power condition $\lambda\geq \iota_k(\alpha)$ harder to satisfy. Thus, poorly estimated \acp{ED} are less likely to transmit at $P_{\max}$, as they contribute more misalignment noise to AirComp aggregation.
  This remains a conditional structural property, since $\lambda^*$ and $\alpha^*$ may also depend on~$\sigma_{h,k}^2$.
\end{remark}

For given $\alpha>0$ and $\lambda\geq 0$, the power allocation $x_k^*(\alpha, \lambda)$ is fully determined by~\Cref{theorem: closed form solution of x_k}.
\begin{proposition}[Case $\lambda^*=0$]
    \label{theorem: lambda eq zero}
    Under \Cref{assumption: strict feasible}, if $\lambda^*=0$, then the primal optimum of $\mathrm{(P2)}$ is the solution of \Cref{lemma: res_lambda_zero}.
\end{proposition}
\begin{proof}
This follows directly: the dual variable $\lambda$ is associated with the sensing constraint, and when $\lambda^*=0$, the sensing constraint is inactive. The problem thus reduces to the case without the sensing constraint, whose solution is given by \Cref{lemma: res_lambda_zero}. Note that when $\lambda^* = 0$, $\lambda^* < \iota_k(\alpha^*)$ for any $k\in\setK$ is equivalent to the condition $\hat{h}_k^2 / \tilde{h}_k^4 \leq   \alpha^* P_{\max}$, which coincides with the condition for $x_k^*$ to be smaller than $P_{\max}\alpha^*$ in \Cref{lemma: res_lambda_zero}. 
\end{proof}

The following theorem characterizes necessary and sufficient conditions for optimal~$(\alpha^*,\lambda^*)$.

\begin{theorem}[Optimality conditions]
\label{th: equivalent_KKT_conditions}
Under \Cref{assumption: strict feasible} and if $\lambda^*>0$, the primal and dual optimum~$((x_k^*)_{k\in\setK}, \alpha^*, \lambda^*)$ of $\mathrm{(P2)}$ can be fully determined as follows:\\
i) For all $k\in\setK$, $x_k^*(\lambda^*, \alpha^*)$ satisfies~\Cref{theorem: closed form solution of x_k}.
\\
ii) 
\begin{equation}
    \sum_{k\in\setK} b_kx_k^*(\lambda^*, \alpha^*) = \eta_{\mathrm{D}} \alpha^*.
    \label{eq: sum_cstr_equality}
\end{equation}
 iii) 
 \begin{equation}
  \hspace{.3cm} 
  \frac{\sigma_n^2 \!+\! \lambda^* \eta_{\mathrm{D}}}{P_{\max}} \!=\!  \sum_{k\in\setK} \left(\!\frac{\hat{h}_k }{\sqrt{P_{\max}\alpha^*}} \!-\! (\tilde{h}_k^2 \!-\! \lambda^* b_k)\! \right)^+\!\!.
 \end{equation}
\end{theorem}
\begin{proof}
In Appendix~\ref{proof: theorem equivalent KKT conditions}.
\end{proof}

The positive-part functions in condition~(iii) make it difficult to analyze the relation between $\alpha^*$ and $\lambda^*$.
Let $\overline{\setS}^{(\lambda^*,\alpha^*)}\subset\setK$ be the set of \acp{ED} that transmit with maximum power.
Then we obtain the following closed-form expression of $\alpha^*$:

\begin{corollary}
\label{th: alpha final condition}
Under the same assumptions as in~\Cref{th: equivalent_KKT_conditions}, given the active set $\overline{\setS}^{(\lambda^*, \alpha^*)} = \{k\in\setK\mid x_k^* = \alpha^* P_{\max}\}$, we have the following expression for $\alpha^*=\alpha^*_{(iii)}(\lambda^*)$ given~$\lambda^*$:
\begin{equation}
\alpha^*_{(iii)}(\lambda^*)\! \defeq \! \frac{1}{P_{\max}}\! \Bigg(\! \frac{\sum_{k\in\overline{\setS}^{(\lambda^*, \alpha^*)}}\hat{h}_k }{\frac{\sigma_n^2 + \lambda^*\! D^{(\lambda^*\!, \alpha^*\!)}}{P_{\max}} \!+\! \sum_{k\in\overline{\setS}^{(\lambda^*\!, \alpha^*)}} \! \tilde{h}_k^2 } \!
\Bigg)^2 \!,
\label{eq: alpha^* expression}
\end{equation}
with $D^{(\lambda^* , \alpha^* )} =\eta_{\mathrm{D}} - \sum_{k\in\overline{\setS}^{(\lambda^* , \alpha^* )}} b_kP_{\max}$.
In addition, $D^{(\lambda^* , \alpha^* )} > 0$, and $\emptyset \subsetneq \overline{\setS}^{(\lambda^*, \alpha^*)} \subsetneq \setK $.

\end{corollary}

\begin{proof}
    The idea is to express $\alpha^*$ as a function of $\lambda^*$.
    
By isolating the terms involving $\alpha^*$:
\begin{multline}
\frac{\sum_{k\in\overline{\setS}^{(\lambda, \alpha)}}\hat{h}_k }{\sqrt{P_{\max}\alpha^*}} \!=\! \frac{\sigma_n^2 \!+\! \lambda \eta_{\mathrm{D}}}{P_{\max}} + \hspace{-.2cm}\sum_{k\in\overline{\setS}^{(\lambda, \alpha)}}\hspace{-.12cm}(\tilde{h}_k^2 - \lambda b_k).
\end{multline}
And therefore,
\begin{equation}
\hspace{-.23cm}\sqrt{P_{\! \max}\alpha^*}
\!=\! \frac{\sum\limits_{k\in\overline{\setS}^{(\lambda, \alpha)}}\hat{h}_k }{\frac{\sigma_n^2 + \lambda \eta_{\mathrm{D}}}{P_{\max}} \!+ \hspace{-.4cm} \sum\limits_{k\in\overline{\setS}^{(\lambda, \alpha)}}\hspace{-.25cm}  (\tilde{h}_k^2 \!-\! \lambda b_k)}.
\end{equation}

Denoting $D^{(\lambda, \alpha)}$ as in the proposition statement, $\alpha^*$ has the expression of $\alpha^*_{(iii)}(\lambda^*)$.
The quantity  $D^{(\lambda, \alpha)}$ must be strictly positive; otherwise, the sensing constraint cannot be satisfied with equality as required by \emph{ii)}.

$\overline{\setS}^{(\lambda^*, \alpha^*)} \neq \emptyset$; otherwise, no \ac{ED} transmits at maximum power, so $\rho_k^*=0$ for all $k\in\setK$, and condition~\eqref{KKT: grad_alpha} cannot hold.
$\overline{\setS}^{(\lambda^*, \alpha^*)} \neq \setK$; otherwise, all \acp{ED} transmit at maximum power, and equality~\eqref{eq: sum_cstr_equality} must hold by complementary slackness of $\lambda^*>0$, which contradicts the strict feasibility condition in \Cref{assumption: strict feasible}.
\end{proof}
With~$\lambda^*=0$, we again obtain the form of~$\alpha^*$ in the case without sensing in \Cref{lemma: res_lambda_zero}. 
Having expressed $\alpha^*$ in terms of $\lambda^*$, it remains to determine~$\lambda^*$.
The following corollary provides the necessary conditions for the optimal~$\lambda^*$.

\begin{corollary}[Optimal Condition on $\lambda^*$]
\label{th: lambda final condition}
   Under the same assumptions as in \Cref{th: equivalent_KKT_conditions}, let $\overline{\setS}^{(\lambda^*, \alpha^*)}$ and $D^{(\lambda^*, \alpha^*)}$ be defined as in \Cref{th: alpha final condition}.
$\lambda^*$ satisfies the following equality:
    \begin{multline}
  \sum_{k\in\setK\backslash \overline{\setS}^{(\lambda^*, \alpha^*)}} b_k\frac{\hat{h}_k^2 }{(\tilde{h}_k^2 - \lambda^*b_k)^2}= D^{(\lambda^*, \alpha^*)}\alpha_{(iii)}^*(\lambda^*).
\label{eq: lambda final condition}
\end{multline}
In addition, $\lambda^*$ satisfies~\eqref{eq: equivalence for lambda expression} for all $k\in\setK\backslash \overline{\setS}^{(\lambda^*, \alpha^*)}$, i.e., $\lambda^* < \min_{k\in\setK\backslash \overline{\setS}^{(\lambda^*, \alpha^*)}} \{ \tilde{h}_k^2/b_k(1-\hat{h}_k/(\tilde{h}_k^2\sqrt{\alpha^* P_{\max}})) \}$.
\end{corollary}

\begin{proof}
In Appendix~\ref{proof: theorem lambda final condition}.
\end{proof}

At this point, if the active set $\overline{\setS}^{(\lambda^*, \alpha^*)}$ is known, the dual optimum $\lambda^*$ can be readily obtained via \Cref{th: lambda final condition}.
A zero-finding algorithm (e.g., bisection or Newton's method~\cite{Book_numericalAnalysis}) can be employed to find the unique $\lambda^*$ via~\eqref{eq: lambda final condition}, since the LHS is monotonically increasing and the RHS is monotonically decreasing in $\lambda^*$.

The following theorem establishes that if the set of form~$\overline{\setS}^{(\overline{\lambda}, \overline{\alpha})}$ can be found for which there exists $\overline{\lambda}$ satisfying the conditions in \Cref{th: lambda final condition}, then the corresponding $\overline{\lambda}$ and $\alpha^*_{(iii)}(\overline{\lambda})$ are optimal.

\begin{theorem}
\label{theorem: sufficient condition optimality set form}
Under the same assumptions as in \Cref{th: equivalent_KKT_conditions}, for $\lambda,\alpha>0$, define the active set $\overline{\setS}^{(\lambda, \alpha)} \defeq \{k\in\setK\mid x_k^*(\lambda,\alpha) = \alpha P_{\max}\}$, with $x_k^*(\lambda,\alpha)$ satisfying \Cref{theorem: closed form solution of x_k}. 
If the condition in \Cref{th: lambda final condition} is satisfied with $\overline{\lambda}$ and $\overline{\alpha}=\alpha^*_{(iii)}(\overline{\lambda})$ that preserves the set~$\overline{\setS}^{(\overline{\lambda}, \overline{\alpha})}=\overline{\setS}^{(\lambda, \alpha)}$, then the resulting $\overline{\lambda}$ is the dual optimum, and $\overline{\alpha}$ and the corresponding $x_k^*(\overline{\lambda},\overline{\alpha})$ are the primal optimum of $\mathrm{(P2)}$.
\end{theorem}

\begin{proof}
    The theorem follows from the sufficiency of the KKT conditions for optimality in convex optimization under Slater's constraint qualification. If such a point exists, then all KKT conditions hold, and optimality follows.
\end{proof}

The remaining task is to characterize the active set $\overline{\setS}^{(\lambda, \alpha)}$ and to determine whether there exists such $\lambda$ satisfying the condition:
\begin{equation*}
   \mathrm{(P2)} \iff \textrm{find }\overline{\setS}^{(\lambda, \alpha)} \textrm{ such that \Cref{th: lambda final condition} holds.}
   \label{eq: problem equivalence of finding sets}
\end{equation*}

\begin{theorem}
\label{th: alpha ranges}
The ordering of $\{\iota_k(\alpha)\}_{k\in\setK}$ depends on $\alpha$ and admits at most $K(K- 1)/2 + 1$ distinct orderings. The transition points (in $\alpha\geq 0$) are characterized by:
\begin{equation}
\hspace{-.3cm}\Big(\! \forall (k,j)\! \in\! \Big\{\!(m,n)\!\in\! \setK^2 \!\mid \! \frac{\tilde{h}_m^2}{b_m}\! \neq\! \frac{\tilde{h}_n^2}{b_n}\Big\} \Big)  \sqrt{\!\alpha P_{\!\max}} \!=\!  \frac{ \frac{\hat{h}_k}{b_k} \!-\! \frac{\hat{h}_j}{b_j}}{\frac{\tilde{h}_k^2}{b_k} \!-\! \frac{\tilde{h}_j^2}{b_j}}.
\end{equation}
\end{theorem}
\begin{proof}
In Appendix~\ref{proof: theorem alpha ranges}.
\end{proof}

Let $\alpha_1<\dots < \alpha_{N_{\alpha}}$ with $N_{\alpha}\leq K(K- 1)/2 $ denote, in ascending order, the valid distinct transition points defined in \Cref{th: alpha ranges} . We set $\alpha_0 =0$ and $\alpha_{N_{\alpha}+1} = +\infty$. Then the ordering of $\{\iota_k(\alpha)\}_{k\in\setK}$ is fixed for any $\alpha\in [\alpha_n, \alpha_{n+1})$.
\begin{theorem}
    \label{th: sets depending on n and i}
Let $n\in\{0,\ldots,N_{\alpha}\}$. Denote $\pi_n$ as the permutation of $\setK$ such that:
\begin{equation}
\iota_{\pi_n(1)}(\alpha) \leq \iota_{\pi_n(2)}(\alpha) \leq \cdots \leq \iota_{\pi_n(K)}(\alpha),
\end{equation}
for $\alpha\in(\alpha_n, \alpha_{n+1})$.
For any rank index $i\in\{1,\ldots,K-1\}$ and $\lambda\in [\iota_{\pi_n(i)}(\alpha), \iota_{\pi_n(i+1)}(\alpha)) $,
the active set~$\overline{\setS}^{(\lambda, \alpha)}$ is determined by:
\begin{equation}
\overline{\setS}^{(\lambda, \alpha)} = \{\pi_n(1),\ldots,\pi_n(i)\} \defeq \overline{\setS}_{n, i}.
\end{equation}
\end{theorem}

\begin{theorem}[Optimal Condition on $(n,i)$]
    \label{th: algorithm core theorem}
    Under the same assumptions as in \Cref{th: equivalent_KKT_conditions}, fix $n\in\{0, \ldots, N_{\alpha} \}$, the corresponding permutation $\pi_n$, and a rank index $i\in\{1,\ldots,K-1\}$.
    Define $\overline{\setS}_{n, i}$ as in the previous theorem and
    \begin{equation}
        D_{n,i} \defeq \eta_{\mathrm{D}} - P_{\max}\sum_{k\in \overline{\setS}_{n,i}} b_k,
        \label{eq: D_ni}
    \end{equation}
    as well as
    \begin{equation}
        \alpha^{(iii)}_{n,i}(\lambda)\! \defeq \! \frac{1}{P_{\max}}\Bigg(\! \frac{\sum_{k\in\overline{\setS}_{n,i}}\hat{h}_k }{\frac{\sigma_n^2 + \lambda D_{n,i}}{P_{\max}} + \sum_{k\in\overline{\setS}_{n,i}} \tilde{h}_k^2 } \!
\Bigg)^2.
    \end{equation}
    There exists an index pair $(n,i)$ such that $D_{n,i}>0$ and the following condition holds. In this case, the set associated with pair $(n,i)$ corresponds to the optimal set $\overline{\setS}_{n,i}= \overline{\setS}^{(\lambda^*, \alpha^*)}$ for some dual optimum $\lambda^*>0$ and primal optimum $\alpha^*$:
    \\
    Define the (within-interval) increasing and decreasing functions, respectively,
    \begin{align}
       \hspace{-.2cm} \mathrm{LHS}_{n,i}(\lambda) &\defeq \sum_{k\in\setK\backslash \overline{\setS}_{n,i}} b_k\frac{\hat{h}_k^2 }{(\tilde{h}_k^2 - \lambda b_k)^2},\\
        \hspace{-.2cm}  \mathrm{RHS}_{n,i}(\lambda) &\defeq \frac{D_{n,i}}{P_{\max}}\Bigg(\frac{\sum_{k\in\overline{\setS}_{n,i}}\hat{h}_k }{\frac{\sigma_n^2 + \lambda D_{n,i}}{P_{\max}} \!+\! \sum_{k\in\overline{\setS}_{n,i}} \tilde{h}_k^2 } \Bigg)^2.
    \end{align}
    There exists $\lambda^* \in \Big(0, \min\limits_{k\in\setK\backslash\overline{\setS}_{n,i}} \Big\{ \frac{\tilde{h}_k^2}{b_k}(1- \frac{\hat{h}_k}{\tilde{h}_k^2\sqrt{\alpha_{n+1}P_{\max}}}) \Big\} \Big)$ s.t. the following condition holds
    \begin{equation}
        \begin{cases}
            \mathrm{LHS}_{n,i}(\lambda^*) = \mathrm{RHS}_{n,i}(\lambda^*)\\
            \alpha^{(iii)}_{n,i}(\lambda^*) \in [\alpha_n, \alpha_{n+1}]\\
            \lambda^* \in [\iota_{\pi_n(i)}(\alpha^{(iii)}_{n,i}(\lambda^*)), \iota_{\pi_n(i+1)}(\alpha^{(iii)}_{n,i}(\lambda^*))).
        \end{cases}
        \label{eq: conditions to hold alpha lambda (crossing point)}
    \end{equation}
\end{theorem}

\begin{proof}
The existence of such an index pair follows from the existence of an optimal solution of $\mathrm{(P1)}$ by \Cref{theorem: solution existence} and its equivalent problem formulations. The necessity of the optimal solution form is provided by \Cref{th: lambda final condition}, and its sufficiency is provided by \Cref{theorem: sufficient condition optimality set form}.

Note that the interval $[\iota_{\pi_n(i)}(\alpha^{(iii)}_{n,i}(\lambda^*)), \iota_{\pi_n(i+1)}(\alpha^{(iii)}_{n,i}(\lambda^*)))$ may be empty for two reasons:
\\
i) if for any $\alpha\in(\alpha_n, \alpha_{n+1})$, $\iota_{\pi_n(i)}(\alpha) = \iota_{\pi_n(i+1)}(\alpha)$.  This happens for pairs $(k,j)\in\setK^2$ that have the same $(\hat{h}_k, \tilde{h}_k^2,b_k) = (\hat{h}_j, \tilde{h}_j^2,b_j)$. In this case, the corresponding $\lambda$-interval is empty, so no $\lambda$ can satisfy the condition; therefore, this index pair $(n,i)$ can be ignored.
\\
ii) if $\alpha^{(iii)}_{n,i}(\lambda^*) = \alpha_n$ or $\alpha^{(iii)}_{n,i}(\lambda^*) = \alpha_{n+1}$, then we are at a changing point between indices, and the interval of these index pairs is empty.
\end{proof}

Therefore, such a set must exist, and once found, it is guaranteed to be optimal.

\subsection{Algorithm Development}

Building upon the theoretical characterization in~\Cref{th: algorithm core theorem}, we now develop~\Cref{algo: given client selection} to solve $\mathrm{(P2)}$. The key insight from~\Cref{th: algorithm core theorem} is that the optimal active set $\overline{\setS}^{(\lambda^*, \alpha^*)}$ can be determined by systematically searching through all $K^2(K- 1)/2$ possible set configurations characterized by the index pair $(n, i)$.
The algorithm exploits the structure established in~\Cref{th: alpha ranges} and~\Cref{th: sets depending on n and i}: the ordering of $\{\iota_k(\alpha)\}_{k\in\setK}$ changes at most $K(K- 1)/2$ times as $\alpha$ varies; for each interval $[\alpha_n, \alpha_{n+1})$, the permutation $\pi_n$ remains fixed, and the active set $\overline{\setS}_{n,i}$ is uniquely determined by the rank index~$i$.

The algorithm proceeds by iterating over all $\alpha$-intervals indexed by~$n$ and $\lambda$-interval indexed by~$i$. Within each interval, the conditions from~\Cref{th: algorithm core theorem} are checked: $\lambda^*$ must satisfy $\mathrm{LHS}_{n,i}(\lambda^*) = \mathrm{RHS}_{n,i}(\lambda^*)$ for a $\lambda^*\in[\iota_{\pi_n(i)}(\alpha), \iota_{\pi_n(i+1)}(\alpha))$. Since $\mathrm{LHS}_{n,i}(\lambda)$ is increasing and $\mathrm{RHS}_{n,i}(\lambda)$ is decreasing in $\lambda$ within the feasible region, a root-finding~\cite{Book_numericalAnalysis} method efficiently finds the \emph{unique crossing point} $\lambda^*$.
For each candidate solution, the algorithm verifies that the computed $\alpha^{(iii)}_{n,i}(\lambda^*)$ lies in the correct $\alpha$-interval and that $\lambda^*$ satisfies the required $\lambda$-range constraints. The resulting valid candidate yields the optimal solution.
The overall complexity of~\Cref{algo: given client selection} is $\mathcal{O}(|\setK|^3)$. Sorting the $\mathcal{O}(|\setK|^2)$ transition values incurs $\mathcal{O}(|\setK|^2 \log(|\setK|))$ complexity. The algorithm then iterates over $\mathcal{O}(|\setK|^2)$ $\alpha$-intervals, checking $\mathcal{O}(|\setK|)$ candidate solutions per interval. Since the logarithmic term is dominated, the overall complexity takes the stated form. 
Since the algorithm has such analytical characterization, all the operations may be transformed into matrix/tensor operations, resulting in further acceleration.

\begin{algorithm}
\small
\DontPrintSemicolon
\SetAlgoNoEnd
\caption{Proposed Optimal Algorithm for $\mathrm{(P2)}$}
\label{algo: given client selection}
\SetAlgoLined
\KwIn{$\setK$, $\{\hat{h}_k\}_k$, $\{\sigma_{h,k}^2\}_k$, $\{b_k\}_k$, $\sigma_n^2$, $\eta_{\mathrm{D}}$, $P_{\max}$.}
\textbf{Preprocessing:}
Check the strict feasibility of the problem. Check if \Cref{lemma: res_lambda_zero} solution is already feasible. \textbf{Return} if it is the case.
\\
Compute all $\alpha$-interval boundaries to obtain~$\alpha_n$.\;
\For{$n = 0, \ldots, K(K-1)/2$}{
    Determine permutation $\pi_n$ for $\alpha \in [\alpha_n, \alpha_{n+1})$.\;
    \For{$i = 1,\ldots,K-1$}{
        Construct $\overline{\setS}_{n,i}$ (\Cref{th: sets depending on n and i}).
        Compute $D_{n,i}$ as in \eqref{eq: D_ni}.\;
        \lIf{$D_{n,i} \leq 0$}{\textbf{continue}}
        Root finding algorithm to find $\lambda^*$ s.t.\ $\mathrm{LHS}_{n,i}(\lambda^*) = \mathrm{RHS}_{n,i}(\lambda^*)$; compute $\alpha^{(iii)}_{n,i}(\lambda^*)$.\;
        \lIf{Conditions~\eqref{eq: conditions to hold alpha lambda (crossing point)} hold}{\textbf{return} $\overline{\setS}_{n,i}$, $\lambda^*$, $\alpha^{(iii)}_{n,i}(\lambda^*)$}
    }
}
\KwOut{$\lambda^*$, $\alpha^*\defeq \alpha^{(iii)}_{n,i}(\lambda^*)$, $x_k^*(\lambda^*,\alpha^*)$.}
\end{algorithm}

\subsection{Robustness Analysis of \Cref{algo: given client selection}}
The root-finding step of~\Cref{algo: given client selection} has to be computed numerically and therefore inevitably introduces some error, although the error can be made arbitrarily small by increasing the number of iterations (using bisection or Newton's method~\cite{Book_numericalAnalysis}). Since the algorithm consists of identifying a feasible pair $(n,i)$ (a set $\overline{\setS}_{n,i}$) that satisfies the condition~\eqref{eq: conditions to hold alpha lambda (crossing point)}, no inherent solution continuity property exists. Therefore, it is important to establish its robustness to such numerical inaccuracies, i.e., a small perturbation of the optimal $\lambda^*$ does not lead to a totally different pair $(n,i)$, and a small perturbation of a non-optimal pair $(n',i')$ would not make it satisfy~\eqref{eq: conditions to hold alpha lambda (crossing point)}.

In Appendix~\ref{proof: sensitivity analysis}, we have established that the algorithm remains robust when the root-finding algorithm returns the root estimate $\hat{\lambda}$ such that $\mathrm{LHS}_{n,i}(\hat{\lambda}) \geq \mathrm{RHS}_{n,i}(\hat{\lambda})$.

\section{Simulation Results}
\subsection{General Simulation Settings}
The system parameters are set as follows: $K=20$, $P_{\max}=\qty{23}{\deci\belm}$, $\sigma_n^2=\sigma_d^2 = \qty[per-mode = symbol]{-80}{\deci\belm}$;
the number of receive antennas for sensing is $N_r=4$ and the RCS is~$\zeta=0.7$; the false-alarm probability and detection threshold are $p_{\mathrm{FA}}=10^{-2}$ and $p_{\mathrm{D,th}}= 0.99$, respectively; the signal duration for detection is $T=8$~\cite{Cheng_XuJie_CoordinatedTransmitBeamforming_ISAC_2024}.  The path loss exponent for communication is 2.7. By default, perfect CSI is assumed, i.e., $\sigma_{h,k}=0$. All reported MSE values correspond to the scale- and dimension-normalized MSE, \tMSEOTA, defined in~\eqref{eq: MSE_OTA}.

The FL task is handwritten digit classification on the MNIST dataset~\cite{lecun-mnist-2010} and color image classification on the CIFAR-10 dataset~\cite{cifar}, with data distributed i.i.d. and non-i.i.d. (each \ac{ED} holds data from two classes) while keeping the amount of data equally distributed across the~$K$ \acp{ED}.
A three-layer fully connected neural network serves as the learning model for MNIST, and ResNet-20~\cite{resnet_2016} for CIFAR-10. A learning rate of 0.1 and a local batch size of 64 are used for both networks. Under MNIST, one local iteration is performed before aggregation, whereas under CIFAR-10 with ResNet-20, three local epochs are performed before AirComp aggregation.

\subsection{Optimization Evaluation}

We evaluate the proposed optimization algorithm under diverse system realizations. To efficiently sample topologies likely to satisfy the strict-feasibility condition in \Cref{assumption: strict feasible}, ED and target placements are generated as follows: 3 \acp{ED} are randomly drawn within a \qty{30}{\meter} radius disc of the target to ensure feasibility of the detection constraint, the remaining \ac{ED} are uniformly sampled within a disc of radius \qty{300}{\meter}. The target distance takes the values \qtylist{100;200;300;400;500}{\meter}. For each target location, 100 random realizations of feasible \ac{ED} positions are generated, and for each feasible position realization, 10 independent small-scale fading draws are conducted.

\begin{table*}[t]
	\centering
	\caption{Relative optimality gap distribution across different methods with respect to results from Algorithm~1.}
	\label{tab:relative_error_distribution}
	\setlength{\tabcolsep}{6pt}
	\renewcommand{\arraystretch}{0.6}

	\begin{tabular}{l c c c c c c c c}
		\toprule
		Method & $<\qty{0}{\percent}$ & $<\qty{0.01}{\percent}$ & $<\qty{0.1}{\percent}$ & $<\qty{1}{\percent}$ & $<\qty{10}{\percent}$ & $<\qty{100}{\percent}$ & $<\qty{1000}{\percent}$ & $\geq \qty{1000}{\percent}$ \\
		\midrule
		IPOPT (P1) with Scaling & \textbf{\qty{0.000}{\percent}} & \qty{0.990}{\percent} & \qty{18.826}{\percent} & \qty{24.197}{\percent} & \qty{27.989}{\percent} & \qty{44.780}{\percent} & \qty{75.390}{\percent} & \qty{24.610}{\percent} \\
		IPOPT (P2) w/o Scaling & \textbf{\qty{0.000}{\percent}} & \qty{6.660}{\percent} & \qty{9.970}{\percent} & \qty{16.078}{\percent} & \qty{27.365}{\percent} & \qty{49.888}{\percent} & \qty{94.845}{\percent} & \qty{5.155}{\percent} \\
		IPOPT (P2) with Scaling & \textbf{\qty{0.000}{\percent}} & \qty{99.973}{\percent} & \qty{99.977}{\percent} & \qty{99.980}{\percent} & \qty{99.988}{\percent} & \qty{99.995}{\percent} & \qty{99.999}{\percent} & \qty{0.001}{\percent} \\
		\bottomrule
	\end{tabular}

\end{table*}

We start by validating the optimality of the proposed method for solving $\mathrm{(P2)}$ (Algorithm~1).
The proposed method is compared with the following baselines:
\begin{itemize}[leftmargin=11pt]
    \item \textit{Solver w/o Scaling}: applying a standard off-the-shelf \ac{NLP} solver, IPOPT~\cite{ipopt_2006}, directly to the convex optimization problem $\mathrm{(P2)}$ without any tailored scaling. In the cases where the solver does not return a feasible point, a feasible heuristic solution (later specified as \textit{greedy Sensing Power}) is used instead.
    \item \textit{Solver with Scaling}: applying the same off-the-shelf solver to $\mathrm{(P2)}$ with problem-specific scaling (variable transformation $y_k = \tilde{h}_k^2x_k$) to mitigate the ill-conditioning of the problem.
    \item \textit{Solver for $\mathrm{(P1)}$ with Scaling}: applying the same off-the-shelf solver to the original non-convex \ac{NLP} problem $\mathrm{(P1)}$ with the same scaling as above, in order to illustrate the importance of the variable transformation that makes the optimization problem convex~\eqref{eq: variable transformation pk xk}.
\end{itemize}

The relative performance of different power allocation methods compared to the optimal solution from Algorithm~1 is shown in Table~\ref{tab:relative_error_distribution}. We first observe that the relative gap is never negative, confirming that the baselines never outperform the proposed algorithm and suggesting that the proposed algorithm attains the optimum with high numerical precision. 
After correctly scaling $\mathrm{(P2)}$, \qty{99.973}{\percent} of the solutions are numerically identical ($<\qty{0.01}{\percent}$) between the proposed method and the off-the-shelf solver benchmark, however rare larger gaps still occur for this case.
By contrast, the corresponding benchmark without scaling on $\mathrm{(P2)}$ often cannot reach the optimum, and applying the same scaling to the non-convex formulation $\mathrm{(P1)}$ also leads to inferior results.
These results highlight the importance of the proposed convex variable transformation~\eqref{eq: variable transformation pk xk} and the numerical stability of the proposed analytically-characterized algorithm: off-the-shelf solvers can be sensitive to formulation and scaling, whereas the proposed algorithm achieves the optimum without requiring such tuning.

\begin{figure}[t]
    \centering
    
    \subfloat[Average run time\label{fig: average time}]{\includegraphics[width=0.49\linewidth]{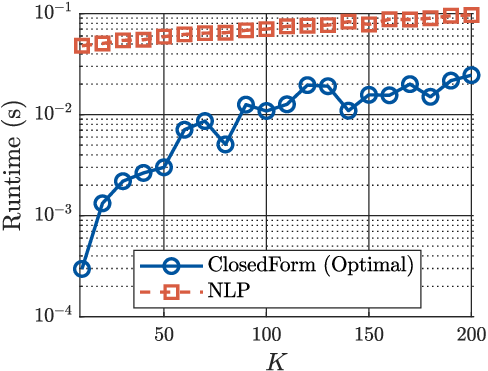}}
    \hfill
    \subfloat[Runtime ratio histogram\label{fig: time distribution}]{\includegraphics[width=0.49\linewidth]{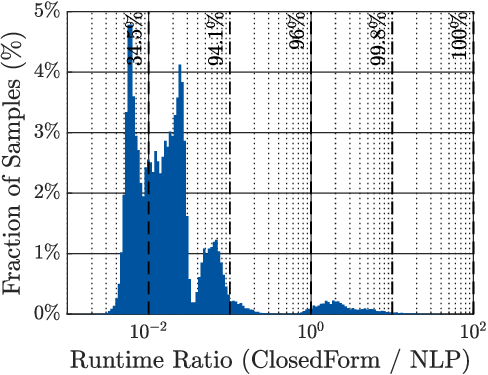}}
    \caption{Run time comparison between off-the-shelf NLP solver and proposed optimal solver.}
    \label{fig: time_gap_NLP}
\end{figure}
We next compare the runtime of the proposed method with the off-the-shelf solver (IPOPT with scaling on $\mathrm{(P2)}$). The average runtime across all scenarios is shown in Fig.~\ref{fig: average time}, and the distribution of runtime ratios (proposed method time divided by off-the-shelf solver time) is shown in Fig.~\ref{fig: time distribution}. The proposed method achieves a significant reduction in runtime, with an average runtime of about \qty{0.001}{\second} compared to about \qty{0.01}{\second} for the off-the-shelf solver, corresponding to a 10$\times$ speedup. Due to the potentially large number of $\alpha$-interval searches (although still polynomial complexity), the runtime may vary. The runtime ratio histogram shows that in \qty{34.5}{\percent} of cases, the proposed method achieves a 100$\times$ speedup, in \qty{94.1}{\percent} of cases it achieves at least a 10$\times$ speedup, and in \qty{96}{\percent} of cases it is at least faster as the off-the-shelf solver. However, in rare cases (\qty{4}{\percent}), the proposed method can be slower than the off-the-shelf solver; in most of these cases it is not much slower, but in \qty{0.2}{\percent} of cases it may be more than 10$\times$ slower.

\subsection{MSE Performance}
Having verified the optimality of the proposed algorithm, we compare the following baselines with the optimal power control in terms of the achieved MSE at the BS:
\begin{itemize}[leftmargin=11pt]
    \item \textit{Zero-forcing (ZF) scaling}: determine the receive scaling factor $\alpha$ large enough that ZF can be applied for $p_k$ and then compute the associated ZF power control.
    \item \textit{Greedy Sensing Power} (Greedy SP): let the EDs with the largest $b_k$ transmit at $P_{\max}$ until the sensing constraint is satisfied, then perform minimum-MSE power control without the sensing constraint for the remaining EDs.
    \item \textit{No sensing}: solve the MSE minimization without enforcing the sensing constraint, which provides an optimistic lower bound on the achievable performance.
\end{itemize} 

To evaluate the effects of the target location and the relative location of users with respect to the target, we consider a scenario with two \ac{ED} groups as in \Cref{fig: illustration_user_groups}. Each \ac{ED} group is confined within a width of \qty{1}{\meter}. The first \ac{ED} group is located \qty{50}{\meter} away from the BS, the second \ac{ED} group is located at $d_{\mathrm{2nd}}$, and the target is located at $d_{t}$. The distance between the target location and the second \ac{ED} group is therefore $d_{\mathrm{2nd}} - d_t$.
\begin{figure}[t]
    \centering
    \includegraphics[width=0.48\linewidth, trim= 0 0 0 0]{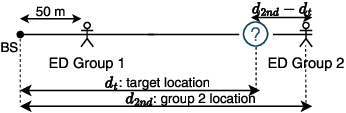}
    \caption{Illustration of the two-user-group scenario.}
    \label{fig: illustration_user_groups}
\end{figure}

We evaluate the MSE performance of different power control methods by varying the second \ac{ED} group location $d_{\mathrm{2nd}}$ and the target distance to the second \ac{ED} group $d_{\mathrm{2nd}} - d_t$. The average results over $40$ random Rayleigh fading realizations are shown in \Cref{fig: MSE vs second_location}. First, the MSE of the optimal design is consistently the lowest across all scenarios, confirming the optimality of the proposed method. The gap relative to \emph{No Sensing} is small for very near-target and very far-target scenarios.
The MSE of ZF is particularly sensitive to the relative position between the target and the second \ac{ED} group; it reaches more than $10^2$ when $d_{\mathrm{2nd}} - d_t = \qty{12}{\meter}$.
This is because ZF depends on the weakest channel strength and is therefore sensitive to deep fades, especially when there is little room for optimization, i.e., here for $d_{\mathrm{2nd}} - d_t = \qty{12}{\meter}$, where the second \ac{ED} group's power almost have to be $P_{\max}$ to guaranteeing system feasibility. 
The \emph{greedy SP}, on the other hand, performs with small gaps to the optimal one in most cases except for the near-target scenario (e.g., at $d_{\mathrm{2nd}} = \qty{70}{\meter}$ and $d_{\mathrm{2nd}}-d_t=\qty{10}{\meter}$), where the sensing constraint is loose and there is more room for optimization. Its performance gap relative to the optimal and \emph{No Sensing} schemes decreases as $d_{\mathrm{2nd}}$ increases. This is because the EDs in the second group tend to transmit at their maximum power under both the optimal and \emph{No Sensing} schemes, which increasingly coincides with the \emph{greedy SP} strategy. Moreover, the impact of suboptimal power allocation diminishes as these EDs move farther away from the BS. 

\begin{figure}[t]
    \centering
    \includegraphics[width=0.57\linewidth, trim=0 -10 0 0]{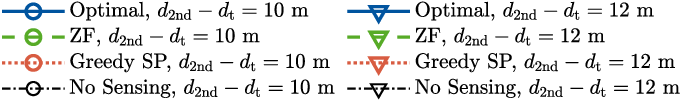}
    
    \includegraphics[width=0.57\linewidth]{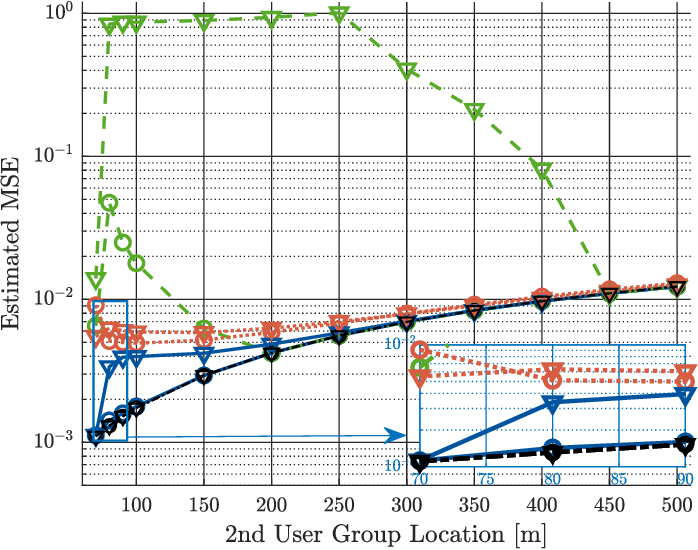}
    \caption{Evaluation of MSE vs. second \ac{ED} group location with different target distances to the second \ac{ED} group.}
    \label{fig: MSE vs second_location}
\end{figure}

We also evaluate in this setting the effects of the detection threshold~$p_{\mathrm{D,th}}$ and the false-alarm probability threshold~$p_{\mathrm{FA}}$ on the proposed algorithm. We set $p_{\mathrm{FA}} = 1-p_{\mathrm{D,th}}$ so that the false-alarm and missed-detection probabilities remain of the same order of magnitude, as is often desirable in practice. The effect of~$p_{\mathrm{D,th}}$ is shown in Fig.~\ref{fig: p_D_effect2}. The MSE ranges from $0.001$ to $0.007$ and the problem becomes infeasible when the sensing requirement is tightened further. 

\begin{figure}[t]
	\centering
	\includegraphics[width=0.55\linewidth]{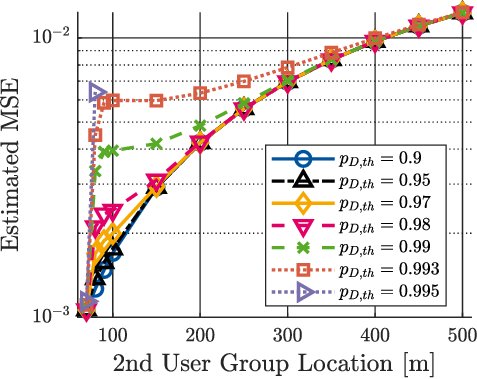}
	\caption{Effect of $p_{\mathrm{D,th}}$ on the MSE for $d_{\mathrm{2nd}}-d_t=12$\,m, with $p_{\mathrm{FA}}=1-p_{\mathrm{D,th}}$.}
	\label{fig: p_D_effect2}
\end{figure}

\subsection{AirComp-FL Performance}

The end-to-end performance and the MSE of AirComp-FL are evaluated in terms of test accuracy and aggregation MSE across two datasets (MNIST and CIFAR-10) and two data distributions specified in section IV. A. (i.i.d. and non-i.i.d.) in Figs.~\ref{fig: mnist 70}-\ref{fig: cifar 300}.
Two network geometries are evaluated: a near-target scenario in which the second \ac{ED} group is located at $d=\qty{70}{\meter}$ from the BS with the sensing target \qty{10}{\meter} from the second \ac{ED} group (where greedy SP performs poorly), and a far-target scenario at $d=\qty{300}{\meter}$ with the target \qty{12}{\meter} from the second \ac{ED} group (where ZF performs especially poorly in terms of MSE). In addition to \emph{Greedy SP}, \emph{ZF}, and \emph{No Sensing}, we also compare with the following methods:
\begin{itemize}[leftmargin=12pt]
	\item \emph{Orthogonal}: Model updates are transmitted over orthogonal resources and are therefore aggregated without distortion at the \ac{BS}. This serves as a performance upper bound.
	\item \emph{Equal Power}: All \acp{ED} transmit at their maximum power, and the receive scaling factor~$\alpha$ is then optimized. This is a feasible solution to~$\mathrm{(P2)}$.
	\item \emph{Weighted-aggregation (WA) AirComp-FL}~\cite{Azimi_OTA_FL_WeightedAggregation_2024}: All \acp{ED} transmit at their maximum power, so this scheme is also feasible for~$\mathrm{(P2)}$. Instead of imposing common normalization factors, each \ac{ED} normalizes its update individually. Based on the per-\ac{ED} normalization factors and the receiver-side CSI, the \ac{BS} optimizes the receive scaling and aggregation weights either to minimize the MSE (\emph{WA MSE})~\cite[Algorithm 1]{Azimi_OTA_FL_WeightedAggregation_2024} or to minimize the mismatch (\emph{WA minMismatch})~\cite[Algorithm 2]{Azimi_OTA_FL_WeightedAggregation_2024}.
\end{itemize}

In all configurations, the proposed optimal power control closely matches the \emph{No Sensing} and \emph{Orthogonal} upper bounds in terms of test accuracy. This indicates that \emph{\ac{Sig-ISCC}-based FL can achieve performance close to the case without aggregation distortion and without a sensing task, meaning that the effect of the sensing constraint can be largely mitigated through optimal joint power allocation and receive scaling}.
The aggregation MSE of the optimal design is consistently the lowest among all methods. Greedy SP, as predicted by Fig.~\ref{fig: MSE vs second_location}, performs closely to the performance upper bound in all cases $d_{\mathrm{2nd}}\!=\!300\,$m and also achieves competitive accuracy when $d_{\mathrm{2nd}}\!=\!70\,$m under i.i.d. data (Fig.~\ref{fig: mnist iid 70} and Fig.~\ref{fig: cifar iid 70}) but exhibits slower convergence or a lower final accuracy in the non-i.i.d. setting (Fig.~\ref{fig: mnist niid 70} and Fig.~\ref{fig: cifar niid 70}), for example requiring about 100 additional communication rounds to reach 50\% test accuracy in Fig.~\ref{fig: cifar niid 70}. 
ZF is particularly sensitive to deep-fade channel realizations, with MSE spikes spanning several orders of magnitude and abrupt accuracy degradations throughout training especially under $d_{\mathrm{2nd}}\!=\!300$\,m in Fig.~\ref{fig: mnist 300} and Fig.~\ref{fig: cifar 300}. Among the stable baselines (i.e., except ZF), \emph{Equal Power} performs the worst overall. It remains competitive only for the simpler MNIST task with $d_{\mathrm{2nd}}\!=\! 70$\,m (Fig.~\ref{fig: mnist 70}), but degrades markedly for the more challenging CIFAR-10 task, especially under non-i.i.d. data (Fig.~\ref{fig: cifar 70}), and deteriorates further in the more spatially diverse geometry with $d_{\mathrm{2nd}}\!=\!300$\,m (Figs.~\ref{fig: mnist 300} and~\ref{fig: cifar 300}). \emph{WA MSE} and \emph{WA minMismatch}, which improve upon \emph{Equal Power} via weighted aggregation optimization, perform substantially better but still exhibit non-negligible gaps, particularly under non-i.i.d. CIFAR-10 settings (Figs.~\ref{fig: cifar niid 70} and~\ref{fig: cifar niid 300}). 

\begin{figure}[t]
    \centering    
     \includegraphics[width=0.85\linewidth, trim= 0 8 0 0]{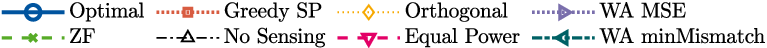}
    \subfloat[MNIST i.i.d.\label{fig: mnist iid 70}]{
        \begin{minipage}{0.48\linewidth}
            \centering
            \includegraphics[width=\linewidth]{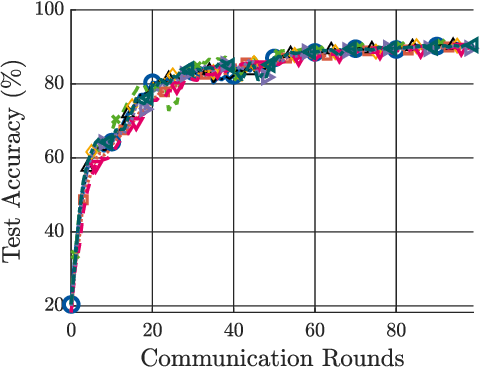} \\
            \includegraphics[width=\linewidth, trim=3 0 0 0]{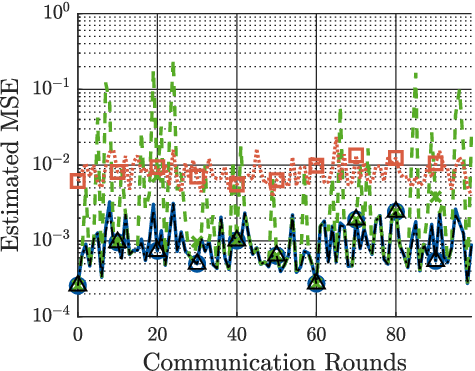}
        \end{minipage}
    }
    \subfloat[MNIST non-i.i.d.\label{fig: mnist niid 70}]{
        \begin{minipage}{0.48\linewidth}
            \centering
            \includegraphics[width=\linewidth]{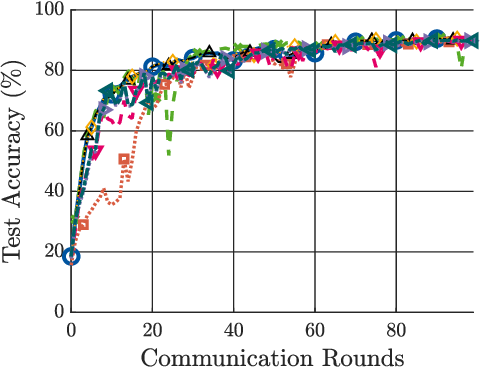} \\
            \includegraphics[width=\linewidth, trim=3 0 0 0]{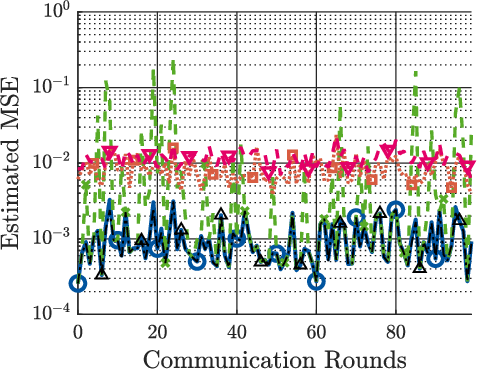}
        \end{minipage}
    }
    \caption{Test accuracy (top) and aggregation MSE (bottom) on MNIST, with the settings ($d_{\mathrm{2nd}}=\qty{70}{\meter}$ and $d_{\mathrm{2nd}} - d_t = \qty{10}{\meter}$).}
\label{fig: mnist 70}
\end{figure}

\begin{figure}[t]
    \centering
      \includegraphics[width=0.85\linewidth, trim= 0 8 0 0]{3_figures_PA_mr1_eps/fig6_legend.eps}
    \subfloat[MNIST i.i.d.\label{fig: mnist iid 300}]{
        \begin{minipage}{0.48\linewidth}
            \centering
            \includegraphics[width=\linewidth]{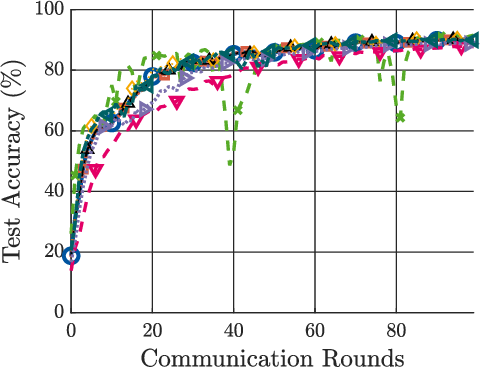} 
            \includegraphics[width=\linewidth, trim=3 0 0 0 ]{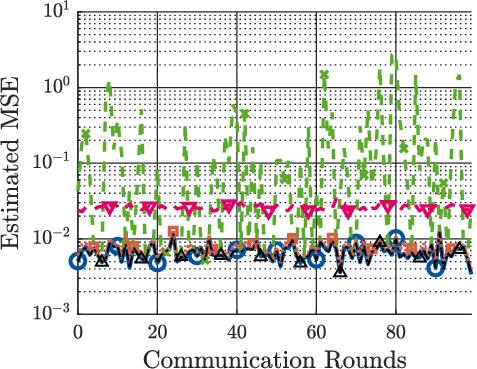}
        \end{minipage}
    }
    \subfloat[MNIST non-i.i.d.\label{fig: mnist niid 300}]{
        \begin{minipage}{0.48\linewidth}
            \centering
            \includegraphics[width=\linewidth]{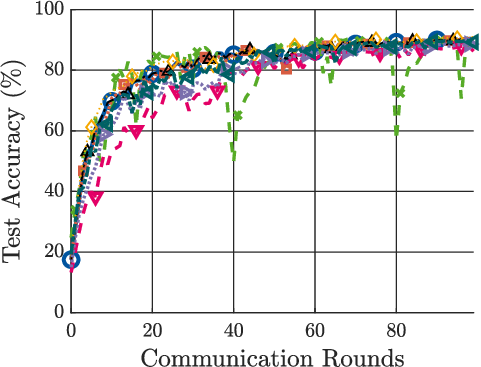} \\
            \includegraphics[width=\linewidth, trim=3 0 0 0]{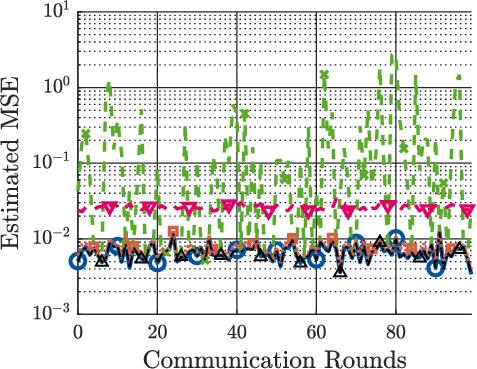}
        \end{minipage}
    }
    \caption{Test accuracy (top) and aggregation MSE (bottom) on MNIST, with the settings ($d_{\mathrm{2nd}}=\qty{300}{\meter}$ and $d_{\mathrm{2nd}} - d_t = \qty{12}{\meter}$).}
    \label{fig: mnist 300}
\end{figure}

\begin{figure}[t]
    \centering
     \includegraphics[width=0.85\linewidth, trim= 0 8 0 0]{3_figures_PA_mr1_eps/fig6_legend.eps}
    \subfloat[CIFAR i.i.d.\label{fig: cifar iid 70}]{
        \begin{minipage}{0.48\linewidth}
            \centering
            \includegraphics[width=\linewidth]{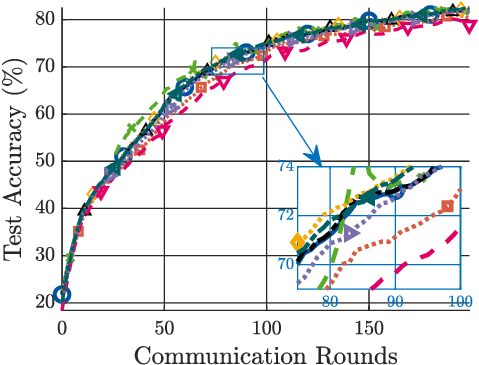} \\
            \includegraphics[width=\linewidth, trim=4 0 0 0]{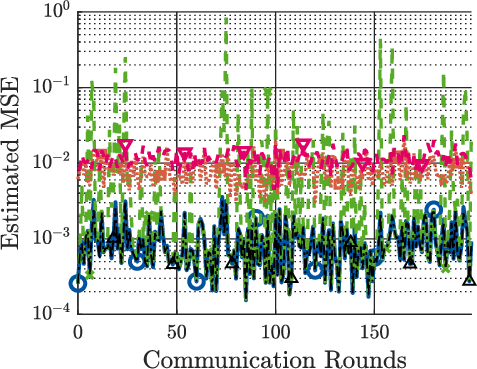}
        \end{minipage}
    }
    \subfloat[CIFAR non-i.i.d.\label{fig: cifar niid 70}]{
        \begin{minipage}{0.48\linewidth}
            \centering
            \includegraphics[width=\linewidth]{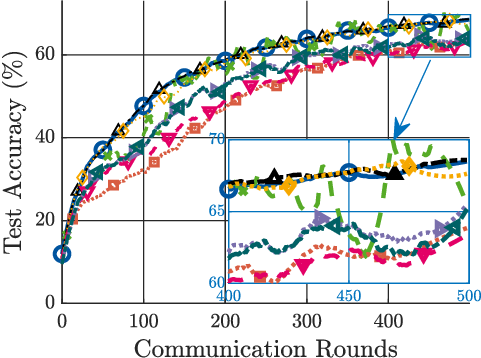} \\
            \includegraphics[width=\linewidth, trim=4 0 0 0]{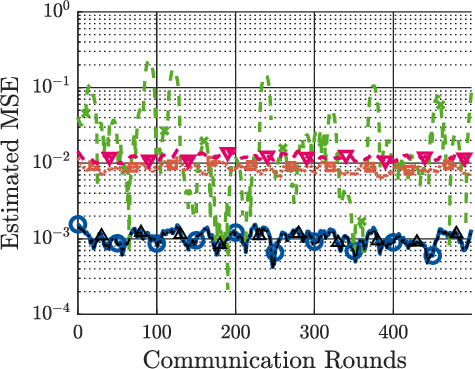}
        \end{minipage}
    }
    \caption{Test accuracy (top) and aggregation MSE (bottom) on CIFAR-10, with the settings ($d_{\mathrm{2nd}}=\qty{70}{\meter}$ and $d_{\mathrm{2nd}} - d_t = \qty{10}{\meter}$).}
    \label{fig: cifar 70}
\end{figure}

\begin{figure}[t]
    \centering
    \includegraphics[width=0.85\linewidth, trim= 0 8 0 0]{3_figures_PA_mr1_eps/fig6_legend.eps}
    \subfloat[CIFAR i.i.d.\label{fig: cifar iid 300}]{
        \begin{minipage}{0.48\linewidth}
            \centering
            \includegraphics[width=\linewidth]{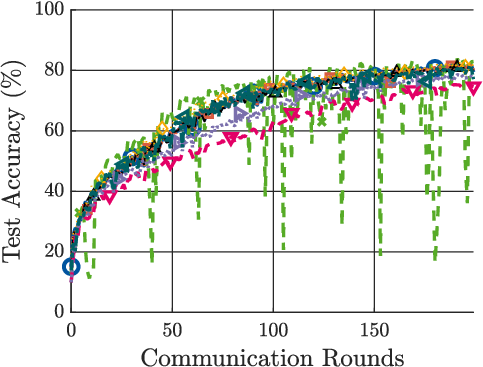} \\
            \includegraphics[width=\linewidth, trim=4 0 0 0]{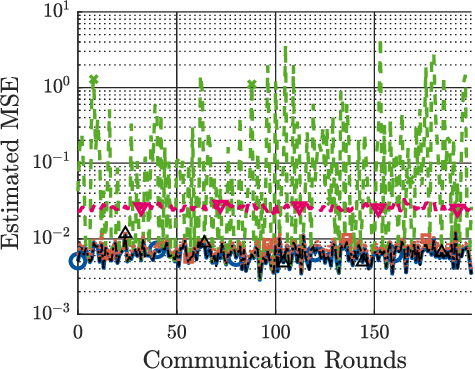}
        \end{minipage}
    }
    \subfloat[CIFAR non-i.i.d.\label{fig: cifar niid 300}]{
        \begin{minipage}{0.48\linewidth}
            \centering
            \includegraphics[width=\linewidth]{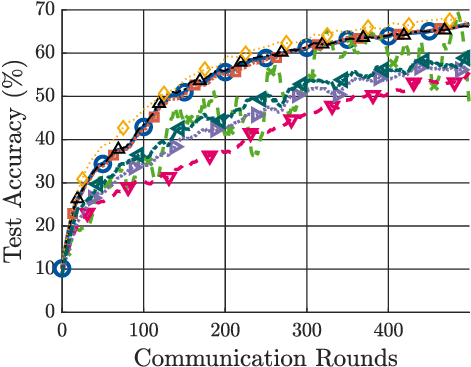} \\
            \includegraphics[width=\linewidth, trim=4 0 0 0]{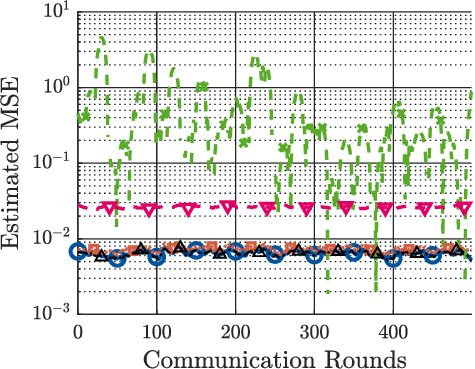}
        \end{minipage}
    }
    \caption{Test accuracy (top) and aggregation MSE (bottom) on CIFAR-10, with the settings ($d_{\mathrm{2nd}}=\qty{300}{\meter}$ and $d_{\mathrm{2nd}} - d_t = \qty{12}{\meter}$).}
    \label{fig: cifar 300}
\end{figure}

\subsection{Imperfect CSI Impact}
We evaluate the impact of the CSI error variance~$\sigma_{h,k}^2$ on the MSE (Fig.~\ref{fig: sigma_h2_mse}) and the FL performance (Fig.~\ref{fig: sigmahk_FL}) by varying the relative CSI error level. The results show that the FL test accuracy starts to degrade noticeably once the relative channel-estimation error reaches about 50\%, at which point the MSE exceeds $0.01$. When the relative error reaches approximately ten times the channel-power level, corresponding to highly noisy channel estimates, the test accuracy drops by about 35\%--45\%.

\begin{figure}[t]
	\centering
	\includegraphics[width=0.5\linewidth]{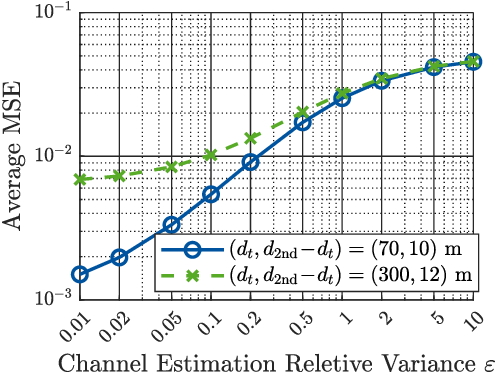}
	\caption{MSE versus the CSI error level $\sigma_{h,k}^2 = h_k^2 \varepsilon$, where $\varepsilon \in \{1\%,\allowbreak 2\%, \allowbreak5\%,\allowbreak 10\%, 20\%, 50\%, 100\%, 200\%, 500\%, 1000\%\}$.}
	\label{fig: sigma_h2_mse}
\end{figure}

\begin{figure}[t]
	\centering
	\includegraphics[width=0.8\linewidth, trim=0 10 0 0]{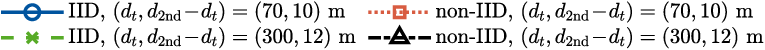}
	\subfloat[MNIST\label{fig: mnist_sigma_h2}]{
		\begin{minipage}{0.48\linewidth}
			\centering
			\includegraphics[width=\linewidth]{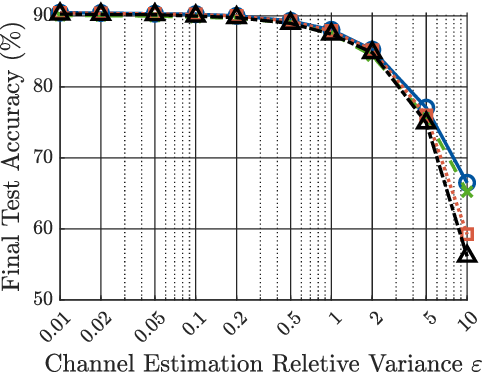} 
		\end{minipage}
	}
	\subfloat[CIFAR-10\label{fig: cifar_sigma_h2}]{
		\begin{minipage}{0.48\linewidth}
			\centering
			\includegraphics[width=\linewidth]{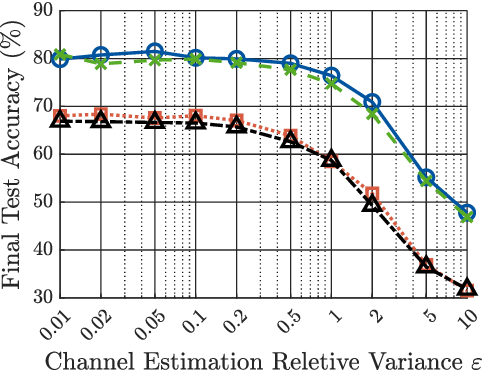} 
		\end{minipage}
	}
	
	\caption{Test accuracy versus the CSI error level $\sigma_{h,k}^2 = h_k^2 \varepsilon$, where $\varepsilon \in \{1\%, \allowbreak 2\%,\allowbreak 5\%,\allowbreak 10\%, 20\%, 50\%, 100\%, 200\%, 500\%, 1000\%\}$.}
	\label{fig: sigmahk_FL}
\end{figure}

\section{Conclusion}
This work investigated the power-control problem of AirFEEL in an uplink Sig-ISCC system subject to a joint target-detection constraint. We formulated the joint transmit-power and receive-scaling design and proposed a variable transformation that transforms the original non-convex problem into an equivalent convex reformulation.
Based on extensive optimality analysis, we developed an analytically characterized polynomial-time optimization algorithm based on root-finding of monotonic differentiable functions. Robustness to root-finding errors has also been proven. Simulation results validate both the optimality and numerical robustness of the proposed design relative to a popular off-the-shelf solver. In addition, the optimal power allocation for \ac{Sig-ISCC} FL preserves FL performance despite the sensing constraint, whereas other non-optimal power allocation frameworks fail to do so.

\bibliographystyle{IEEEtran}
\bibliography{biblio}

\appendices

\section{Proof of \Cref{theorem: solution existence}}
\label{proof: theorem solution existence}

The proof proceeds in two lemmas: first establishing existence under the closed constraint $\alpha\geq 0$, then proving equivalence of the resulting problem with the problem $\alpha>0$.

\begin{lemma}[Existence of minimum under $\alpha\geq0$]
\label{theorem: solution existence in compact set}
When~\eqref{eq: feasible condition} holds, the optimum of the problem~$\mathrm{(P1)}$ with the constraint $\alpha>0$ replaced by $\alpha\geq 0$ exists.
\end{lemma}
\begin{proof}
    
We prove this by showing that the optimum has to lie within a compact set; since the objective function is continuous, a solution exists.
Denote the problem without the sum constraint and restricting $\alpha$ to be greater than a constant value:
\begin{subequations}
\label{pb: for prove existence}
\begin{align}
\min_{ \{p_k\}_{k\in\setK},\alpha}\quad  &
\alpha(\sigma_n^2+\sum_{k\in\setK}\tilde{h}_k^2 p_k) - 2\sum_{k\in\setK}\hat{h}_k\sqrt{\alpha p_k} \\
\mathrm{s.t.} \quad\quad &   0 \leq p_k \leq P_{\max}, \  \forall k \in \setK \label{}\\
& \alpha \geq \max\Big\{\frac{1}{P_{\max}} \max_{k\in\setK}\frac{\hat{h}_k^2}{\tilde{h}_k^4},\  \frac{1}{\sigma_n^2}\sum_{k\in\setK}\frac{\hat{h}_k^2}{\tilde{h}_k^2}\Big\}. \label{cons: alpha minimum}
\end{align}
\end{subequations}
In this case, for a given $\alpha$, the objective is separable in $(p_k)_{k\in\setK}$: with $u_k\defeq\sqrt{\alpha p_k}$, each term $\tilde{h}_k^2u_k^2 - 2\hat{h}_ku_k$ is minimized at $u_k=\hat{h}_k/\tilde{h}_k^2$, i.e., at the regularized channel inversion $\alpha p_k = \hat{h}_k^2/\tilde{h}_k^4$, which is feasible for every $k$ by the first term of~\eqref{cons: alpha minimum}. The resulting optimal objective value can be written as
\begin{equation}
    \alpha\sigma_n^2 - \sum_{k\in\setK}\frac{\hat{h}_k^2}{\tilde{h}_k^2} \geq 0.
\end{equation}
Given the feasible region of~$\alpha$ in~\eqref{cons: alpha minimum}, the optimum of~\eqref{pb: for prove existence} denoted by~$x^*$ is therefore greater than zero.

Now adding the sensing constraint, the problem becomes $\mathrm{(P1)}$ with the additional constraint~\eqref{cons: alpha minimum}. The resulting optimum $y^*$ satisfies $y^*\geq x^* > 0$, since the feasible set is further restricted by the sensing sum-power constraint.

Next, consider $\mathrm{(P1)}$ with the box constraint $\alpha \in [0,\max\Big\{\frac{1}{P_{\max}} \max_{k\in\setK}\frac{\hat{h}_k^2}{\tilde{h}_k^4},\  \frac{1}{\sigma_n^2}\sum_{k\in\setK}\frac{\hat{h}_k^2}{\tilde{h}_k^2}\Big\}]$. An optimum clearly exists since $(p_k)_{k\in\setK}$ and $\alpha$ lie within a compact set and the objective function is continuous. Moreover, the optimal value is less than zero. For $\alpha>0$,
\begin{equation}
\begin{aligned}
& \alpha[  \sigma_n^2+\sum_{k\in\setK}\tilde{h}_k^2p_k ] - 2\sqrt{\alpha}(\sum_{k\in\setK}\hat{h}_k\sqrt{p_k})  < 0
\\
\iff & \sqrt{\alpha}< \frac{2\sum_{k\in\setK}\hat{h}_k\sqrt{p_k}}{\sigma_n^2+\sum_{k\in\setK}\tilde{h}_k^2p_k}.
\end{aligned}
\label{eq: alpha better zero}
\end{equation}

Such an $\alpha>0$ exists whenever $\sum_{k\in\setK}\hat{h}_k\sqrt{p_k}>0$. Setting $p_k=P_{\max}$ for all $k\in\setK$ and choosing $\alpha>0$ satisfying~\eqref{eq: alpha better zero} yields a feasible point (otherwise the problem itself is infeasible) that achieves an objective value strictly smaller than zero.

Combining both cases, optimal~$\alpha$ of~$\mathrm{(P1)}$  has to be within the box constraint. Consequently, the optimum of the original problem lies in a compact set, and the existence of a minimizer is guaranteed. 
\end{proof}

\begin{lemma}[Equivalence $\mathrm{(P1)}$ with $\alpha>0$]
 The problem~$\mathrm{(P1)}$ is equivalent to the problem with $\alpha\geq 0$, i.e. any feasible point~$(\overline{p_k}, \overline{\alpha})$ with $\overline{\alpha}=0$ is not optimal.
 \label{theorem: alpha non zero}
\end{lemma}
\begin{proof}
We now prove that any feasible point with $\overline{\alpha}=0$ is not optimal, thereby establishing that problem~$\mathrm{(P1)}$ with $\alpha>0$ is equivalent to the problem with $\alpha\geq 0$.
When $\overline{\alpha} = 0$, the objective value is 0 for any feasible $\overline{p}_k$.
With the same argument as the previous lemma, there exists a feasible point where $\alpha>0$ and the attained objective is strictly smaller than zero. 
Therefore, no point with $\overline{\alpha}=0$ is optimal.
\end{proof}

Combining both lemmas, we conclude that problem~$\mathrm{(P1)}$ is equivalent to its formulation over the closed convex set with $\alpha\geq 0$, is well-defined, and the optimum is attained. 

\section{Proof of \Cref{lemma: p positive}}
\label{proof: lemma p positive}

Suppose, for the sake of contradiction, that there exists an optimum $(\bx^*, \alpha^*)$ with $x_{k_0}=0$ for some $k_0\in\setK$.
Note that for any $k\in\setK$,
\begin{equation}
\tilde{h}_k^2x_k -2\hat{h}_k \sqrt{x_k} <0
 \iff  \ x_k< \frac{4\hat{h}_k^2}{\tilde{h}_k^4} .
\end{equation}

Construct an alternative point $(\overline{\bx}^*, \alpha^*)$ by setting $\overline{x}_{k_0}=\min\{\hat{h}_{k_0}^2/\tilde{h}_{k_0}^4,P_{\max}\alpha^*\}< \frac{4\hat{h}_{k_0}^2}{\tilde{h}_{k_0}^4}$ and $\overline{\bx}^* = (x_1,\ldots,x_{k_0-1}, \overline{x}_{k_0},x_{k_0+1}\ldots, x_K)$. This point is clearly feasible, and the objective is strictly improved by
\begin{equation}
\tilde{h}_{k_0}^2\overline{x}_{k_0} - 2\hat{h}_{k_0}\sqrt{\overline{x}_{k_0}} < 0.
\end{equation}
Hence, $(\bx^*, \alpha^*)$ cannot be optimal, which contradicts the initial assumption.

\section{Proof of \Cref{th: equivalent_KKT_conditions}}
\label{proof: theorem equivalent KKT conditions}

Part~i) follows from \Cref{theorem: closed form solution of x_k}, and part~ii) is the primal constraint. For iii), taking the partial derivative with respect to $\alpha$ yields:
\begin{equation}
\sigma_n^2+\lambda \eta_{\mathrm{D}}  =  P_{\max}\sum_k \rho_k + \nu.
\end{equation}
By~\Cref{theorem: alpha non zero}, $\alpha^*>0$ therefore $\nu = 0$.
\begin{equation}
    \sigma_n^2+\lambda \eta_{\mathrm{D}}  =  P_{\max}\sum_k \rho_k.
\end{equation}
First, there must exist at least one $\rho_k>0$, which implies that there always exist \acp{ED} transmitting at $x_k^* = P_{\max}\alpha^*$.

We now determine the values of $\rho_k$, noting that $\rho_k$ can be non-zero only when $x_k^* = P_{\max} \alpha^*$.

For $k\in\setK$, from~\eqref{KKT: grad_xk_Sc}, we have:
\begin{equation}
\rho_k = \begin{cases}
    \lambda b_k - \tilde{h}_k^2 + \frac{\hat{h}_k }{\sqrt{P_{\max}\alpha^*}} & \text{ if }  \tilde{h}_k^2\leq \lambda b_k,\\
    \left(\frac{\hat{h}_k }{\sqrt{P_{\max}\alpha^*}} - (\tilde{h}_k^2-\lambda b_k)\right)^+ & \text{ if } \tilde{h}_k^2 > \lambda b_k.
\end{cases}
\end{equation}
In short, it can be written as
\begin{equation}
    \rho_k  =  \left(\frac{\hat{h}_k }{\sqrt{P_{\max}\alpha^*}} - (\tilde{h}_k^2-\lambda b_k)\right)^+.
\end{equation}
We can conclude iii).

\section{Proof of \Cref{th: lambda final condition}}
\label{proof: theorem lambda final condition}

If $\lambda>0$, then equality must hold: $ 
    \sum_{k\in\setK} b_kx_k^* = \eta_{\mathrm{D}} \alpha^*.$
\acp{ED} in $\overline{\setS}^{(\lambda, \alpha)}$ are the exact ones that are transmitting at $P_{\max}\alpha^*$. Others are those $k\in\setK\backslash\overline{\setS}^{(\lambda, \alpha)}$ transmitting at interior power.
We obtain therefore:
\begin{equation}
	\sum_{k\in \setK\backslash\overline{\setS}^{(\lambda, \alpha)}}
	\hspace{-.4cm} b_k \frac{\hat{h}_k^2 }{(\tilde{h}_k^2\!-\! \lambda b_k\!)^2}
	\!=\! \alpha^* \! ( \eta_{\mathrm{D}}\! -\hspace{-.2cm} \sum_{k\in\overline{\setS}^{(\lambda, \alpha)}} \hspace{-.4cm} b_kP_{\max})
	\!\defeq\! \alpha^*\! D^{(\lambda, \alpha)}\hspace{-.3cm}.
\end{equation}

This requires the right-hand side, denoted by $D^{(\lambda, \alpha)}$, to be positive, meaning that the aggregate sensing contribution of the \acp{ED} transmitting at $P_{\max}$ must not exceed the detection threshold.

Replacing the expression of $\alpha^*$ with respect to $\lambda$:
\begin{align*}
\sum_{k\in \setK\backslash\overline{\setS}^{(\lambda, \alpha)}} \hspace{-.35cm}b_k \frac{\hat{h}_k^2 }{(\tilde{h}_k^2-\lambda b_k)^2}
&\!=\!\! \frac{D^{(\lambda, \alpha)} \Big(\sum\limits_{k\in\overline{\setS}^{(\lambda, \alpha)}}\hat{h}_k \Big)^2  }{P_{\!\max}\Big(\frac{\sigma_n^2 \!+\! \lambda \eta_{\mathrm{D}}}{P_{\max}} \!\!+ \hspace{-.41cm} \sum\limits_{k\in\overline{\setS}^{(\lambda, \alpha)}}\hspace{-.37cm} (\tilde{h}_k^2 \!-\! \lambda b_k)\Big)^2}\\
&=\frac{\Big(\sum\limits_{k\in\overline{\setS}^{(\lambda, \alpha)}}\hat{h}_k \Big)^2 D^{(\lambda, \alpha)}}{P_{\!\max}\Big(\frac{\sigma_n^2 + \lambda D^{(\lambda, \alpha)}}{P_{\max}}\! +\hspace{-.1cm} \sum_{k\in\overline{\setS}^{(\lambda, \alpha)}} \hspace{-.1cm} \tilde{h}_k^2 \Big)^2}.
\end{align*}

\section{Proof of \Cref{th: alpha ranges}}
\label{proof: theorem alpha ranges}

For $k\in\setK$, the function $\alpha\mapsto\iota_k(\alpha)\!=\! \frac{\tilde{h}_k^2}{b_k}(1\!-\! \frac{\hat{h}_k}{\tilde{h}_k^2\sqrt{\alpha P_{\max}}})$ is continuous. Therefore, given any ranking of $\{\iota_k\}_{k\in\setK}$, for the ranking to change, it must satisfy
\begin{equation}
\frac{\tilde{h}_k^2}{b_k}(1\!-\! \frac{\hat{h}_k}{\tilde{h}_k^2\! \sqrt{\! \alpha P_{\max}}}) \!=\! \frac{\tilde{h}_j^2}{b_j}(1\!-\!\frac{\hat{h}_j}{\tilde{h}_j^2\!\sqrt{\!\alpha\! P_{\max}}}), \, j\! \in\! \setK, j\!\neq\! k.
\end{equation}
We obtain,
\begin{equation}
\frac{\tilde{h}_k^2}{b_k} - \frac{\tilde{h}_j^2}{b_j} = \frac{1}{\sqrt{\! \alpha P_{\max}}} \Big[\frac{\hat{h}_k }{b_k} \!-\! \frac{\hat{h}_j}{b_j}\Big], \quad j\in\setK, j\neq k.
\end{equation}
If $\tilde{h}_k^2/b_k = \tilde{h}_j^2/b_j$ for some $j\in\setK$, the relative ranking of $k$ and $j$ does not change with $\alpha$, since the right-hand side is either always positive or, if additionally $\frac{\hat{h}_k }{b_k} = \frac{\hat{h}_j}{b_j}$, identically zero, so the two \ac{ED}s share the same ranking regardless of $\alpha$.

Consequently, the transition points of the ranking can only arise among \ac{ED}s with distinct $\tilde{h}_k^2/b_k$ values and can be found by isolating the $\sqrt{\alpha P_{\max}}$ term.

There are at most $|\setK|(|\setK|- 1)/2$ transition points, yielding $|\setK|(|\setK|- 1)/2+1$ possible ranking combinations.

\section{Sensitivity Analysis of~\Cref{algo: given client selection}}
\label{proof: sensitivity analysis}
We show that \Cref{algo: given client selection} is robust to small numerical errors in the root-finding step $\mathrm{LHS}_{n,i}(\lambda)=\mathrm{RHS}_{n,i}(\lambda)$.
Given a pair $(n,i)$ with $D_{n,i}>0$ and denote the denominator in the square term of $\alpha_{n,i}^{(iii)}(\lambda)$ as
\begin{equation}
    q_{n,i}(\lambda)\defeq \frac{\sigma_n^2+\lambda D_{n,i}}{P_{\max}}+\sum_{k\in\overline{\setS}_{n,i}}\tilde{h}_k^2.
\end{equation}

\begin{lemma}[Affine boundary dependence on $\lambda$]
\label{lem: affine_boundary}
For any $j\in\setK$ and any $\lambda>0$ with $q_{n,i}(\lambda)>0$, define
$g_j(\lambda)\defeq \iota_j\!\big(\alpha^{(iii)}_{n,i}(\lambda)\big)$.
Then $g_j$ is affine and strictly decreasing in~$\lambda$:
\begin{equation}
    g_j(\lambda)
    = \frac{\tilde{h}_j^2}{b_j} - \frac{\hat{h}_j}{b_j \sum_{k\in\overline{\setS}_{n,i}}\hat{h}_k}\,q_{n,i}(\lambda),
    \label{eq: g_j_affine}
\end{equation}
and hence, for any $\varepsilon\in\mathbb{R}$,
\begin{equation}
    g_j(\lambda\!+\! \varepsilon)
    \!=\! g_j(\lambda) \!-\! C_{j\!, n\!, i}\varepsilon,
    \,
    C_{j\!,n\!,i}\! \defeq \! \frac{\hat{h}_j D_{n,i}}{b_j P_{\max}\! \sum_{k\in\overline{\setS}_{n,i}}\hspace{-.3cm}\hat{h}_k} \! >\! 0.
    \label{eq: g_j_shift}
\end{equation}
\end{lemma}

\begin{proof}
Since $\sqrt{\alpha^{(iii)}_{n,i}(\lambda)\,P_{\max}} = \Big( {\sum_{k\in\overline{\setS}_{n,i}}\hat{h}_k}\Big) /{q_{n,i}(\lambda)}$, substituting into
$\iota_j(\alpha) = \frac{\tilde{h}_j^2}{b_j}\big(1-\frac{\hat{h}_j}{\tilde{h}_j^2\sqrt{\alpha P_{\max}}}\big)$
gives~\eqref{eq: g_j_affine}. Since $q_{n,i}(\lambda)$ is affine in $\lambda$ with slope ${D_{n,i}}/{P_{\max}}>0$, the shift identity~\eqref{eq: g_j_shift} follows directly.
\end{proof}

\begin{proposition}[Numerical robustness of \Cref{algo: given client selection}]
\label{prop: stability_crossing_test}
Let $F_{n,i}(\lambda)\defeq \mathrm{LHS}_{n,i}(\lambda)-\mathrm{RHS}_{n,i}(\lambda)$ as in~\eqref{eq: conditions to hold alpha lambda (crossing point)}, and let $I\subset\mathbb{R}_+$ be an interval on which all denominators are strictly positive. 
Then $F_{n,i}$ is strictly increasing on~$I$, so its root $\lambda^*$ is unique. Bisection (or Newton's method) can return $\hat{\lambda}\geq\lambda^*$ (i.e., $\varepsilon\defeq\hat{\lambda}-\lambda^*\geq 0$). With sufficiently many iterations, $\varepsilon$ can be made arbitrarily small.
Set $\hat{\alpha}\defeq \alpha^{(iii)}_{n,i}(\hat{\lambda})$.
The three conditions in~\eqref{eq: conditions to hold alpha lambda (crossing point)} are preserved for $(\hat{\lambda},\hat{\alpha})$ in each of the following cases.

\begin{enumerate}[label=(\roman*)]
    \item \textbf{Interior.}\;
    If $\alpha^*\in(\alpha_n,\alpha_{n+1})$ and $\lambda^*\in\big(\iota_{\pi_n(i)}(\alpha^*),\iota_{\pi_n(i+1)}(\alpha^*)\big)$,
    then for sufficiently small $\varepsilon>0$, the conditions hold for $(\hat{\lambda},\hat{\alpha})$.
    
    \item \textbf{Boundary in $\boldsymbol{\alpha}$.}\;
    If $\alpha^*=\alpha_n$ or $\alpha^*=\alpha_{n+1}$, then with $\varepsilon$ sufficiently small, the conditions hold for $(\hat{\lambda},\hat{\alpha})$.

    \item \textbf{$\boldsymbol{\lambda}$-boundary.}\;
    If $\lambda^*=\iota_{\pi_n(i)}(\alpha^*)$, $\hat{\lambda}$ remains in the interval $[\iota_{\pi_n(i)}(\hat{\alpha}), \iota_{\pi_n(i+1)}(\hat{\alpha}))$ for $\varepsilon$ sufficiently small.

    \item \textbf{False-positive protection.}\;
    Suppose $(n,i)$ is not the optimal pair (i.e., one of the conditions does not hold) and let $\bar{\lambda}$ be the exact root of $F_{n,i}$. With sufficiently small $\varepsilon>0$, one of the conditions continues not to hold.
    
    \end{enumerate}

In summary, given the algorithm output $\hat{\lambda}\geq\lambda^*$ with $\varepsilon=\hat{\lambda}-\lambda^*$ arbitrarily small, optimal active sets remain unchanged and non-optimal active sets remain correctly rejected. Therefore, the algorithm is robust to numerical errors in the root-finding step.
\end{proposition}

\begin{proof}
$F_{n,i}$ is clearly strictly increasing on $I$ since the LHS is strictly increasing and the RHS is decreasing in $\lambda$. 

\emph{Case~(i) (interior).}
At $(\lambda^*,\alpha^*)$, all three conditions in~\eqref{eq: conditions to hold alpha lambda (crossing point)} are strict. We first verify that $\hat{\alpha}$ remains in the same $\alpha$-interval. For fixed $(n,i)$, $\alpha_{n,i}^{(iii)}(\lambda)$ is continuously differentiable on the root-search interval $I$. Hence, by the mean value theorem, $|\hat{\alpha}-\alpha^*|
\leq L_{\alpha,n,i}\varepsilon$,
where
\begin{equation}
L_{\alpha,n,i} \defeq \sup_{\lambda\in I}
\left|\frac{d}{d\lambda}\alpha_{n,i}^{(iii)}(\lambda)\right|.
\end{equation}
Since $\alpha^* = \alpha_{n,i}^{(iii)}(\lambda^*)\in(\alpha_n, \alpha_{n+1})$, there exists $\delta_{\alpha}>0$ such that $(\alpha^*-\delta_{\alpha}, \alpha^*+ \delta_{\alpha}) \subseteq (\alpha_n, \alpha_{n+1})$. If $L_{\alpha,n,i}\varepsilon<\delta_{\alpha}$, then $\hat{\alpha}$ remains in $(\alpha_n, \alpha_{n+1})$.

It remains to show that $\hat{\lambda}$ stays in the same $\iota$-interval. Define
\begin{equation}
    \underline{g}(\alpha)\defeq \iota_{\pi_n(i)}(\alpha),
    \qquad
    \overline{g}(\alpha)\defeq \iota_{\pi_n(i+1)}(\alpha).
\end{equation}
Since $\lambda^*\in (\underline{g}(\alpha^*),\overline{g}(\alpha^*))$, there exists $\delta_{\lambda}>0$ such that $(\lambda^*-\delta_{\lambda}, \lambda^*+\delta_{\lambda}) \subseteq (\underline{g}(\alpha^*),\overline{g}(\alpha^*))$. 
Let
\begin{equation}
    L_g \defeq
    \max\left\{
    \sup_{\alpha\in J}|\underline{g}'(\alpha)|,\,
    \sup_{\alpha\in J}|\overline{g}'(\alpha)|
    \right\},
\end{equation}
where $J=(\alpha_n,\alpha_{n+1})$. Since $|\hat{\alpha}-\alpha^*|\leq L_{\alpha,n,i}\varepsilon$, we have
\begin{align}
    \hat{\lambda}-\underline{g}(\hat{\alpha})
    &\geq \lambda^*-\underline{g}(\alpha^*)-(1+L_gL_{\alpha,n,i})\varepsilon,\\
   \hat{\lambda} - \overline{g}(\hat{\alpha})
    &\leq \lambda^* - \overline{g}(\alpha^*)+(1+L_gL_{\alpha,n,i})\varepsilon.
\end{align}
Therefore, if
\begin{equation}
    (1+L_gL_{\alpha,n,i})\varepsilon <\delta_\lambda, \text{ i.e., } \varepsilon < \frac{\delta_\lambda}{1+L_gL_{\alpha,n,i}},
\end{equation}
then $
    \hat{\lambda}\in
    (\underline{g}(\hat{\alpha}),\overline{g}(\hat{\alpha}))$.

\emph{Case~(ii) ($\alpha$-boundary).}
The case $\alpha^*=\alpha_{n+1}$ is trivial since given $\hat{\lambda} > \lambda^*$, $\hat{\alpha} < \alpha_{n+1}$ due to the monotonicity of $\alpha^{(iii)}_{n,i}$.

Assume $\alpha^*=\alpha_n$. By monotonicity and continuity of $\alpha^{(iii)}_{n,i}(\lambda)$ in $\lambda$, sufficiently small $\varepsilon$ keeps $\hat{\alpha}$ in a neighborhood of $\alpha_n$, $(\alpha_{n-1}, \alpha_n)$. At $\alpha_n$, the indices whose $\iota$-values swap ordering satisfy $\iota_{\pi_n(j)}(\alpha_n)=\iota_{\pi_n(j+1)}(\alpha_n)=\iota_{\pi_{n-1}(j)}(\alpha_n)=\iota_{\pi_{n-1}(j+1)}(\alpha_n)$. Given that $\lambda^*\in[\iota_{\pi_n(i)}(\alpha_n),\iota_{\pi_n(i+1)}(\alpha_n))$, the order swap cannot occur at the optimal set $(n,i)$ (otherwise $\lambda^*$ does not exist), i.e., $j\neq i$. Since it involves two neighboring indices, the index swapping does not change the active set $\overline{\setS}_{n,i}$,
so the active set $\overline{\setS}_{n-1,i}$ remains the same as $\overline{\setS}_{n,i}$. With the same continuity argument as in case~(i), $\hat{\lambda}$ remains in the interval $[\iota_{\pi_{n-1}(i)}(\hat{\alpha}), \iota_{\pi_{n-1}(i+1)}(\hat{\alpha}))$ for sufficiently small $\varepsilon$, so the optimal conditions are preserved.

\emph{Case~(iii) ($\lambda$-boundary).}
The interval is open at the right side, so is preserved with (i). We only need to take care of the lower-bound condition.
Suppose $\lambda^*=\iota_{\pi_n(i)}(\alpha^*)$, i.e.,
$\lambda^*=g_{\pi_n(i)}(\lambda^*)$ with $g_j(\lambda)\defeq\iota_j(\alpha^{(iii)}_{n,i}(\lambda))$.
Using \Cref{lem: affine_boundary},
\begin{equation}
    \begin{split}
    \hat{\lambda}-g_{\pi_n(i)}(\hat{\lambda})
    &= (\lambda^*+\varepsilon)-\big(g_{\pi_n(i)}(\lambda^*)-C_{\pi_n(i),n,i}\varepsilon\big)\\
    &= (1+C_{\pi_n(i),n,i})\,\varepsilon \geq 0.
    \end{split}
    \label{eq: lower_bound_preserved}
\end{equation}
Hence $\hat{\lambda}\geq \iota_{\pi_n(i)}(\hat{\alpha})$, so the lower-bound condition is preserved (strict when $\varepsilon>0$). 

\emph{Case~(iv) (false-positive protection).}
Let $(n',i')$ be non-optimal and let $\bar{\lambda}$ be a root of $F_{n',i'}$. One of the conditions of~\eqref{eq: conditions to hold alpha lambda (crossing point)} must fail. With sufficiently small $\varepsilon>0$, $\tilde{\lambda} \defeq \bar{\lambda} + \varepsilon$ should continue to fail one of the conditions. The same logic for (i) and (ii) can be applied similarly to guarantee that $\tilde{\lambda}, \tilde{\alpha}\defeq \alpha_{n',i'}^{(iii)}(\tilde{\lambda})$ continues to fail in the interior and if $\tilde{\alpha}$ is at the $\alpha$-boundaries.  Consider the boundary-failure case
$\bar{\lambda}=g_{\pi_{n'}(i'+1)}(\bar{\lambda})$, where $\bar{\lambda}\notin[g_{\pi_{n'}(i')}(\bar{\lambda}), g_{\pi_{n'}(i'+1)}(\bar{\lambda}))$ but may fall into the set with the error. For $\tilde{\lambda}=\bar{\lambda}+\varepsilon$ with $\varepsilon>0$, \Cref{lem: affine_boundary} gives
\begin{equation}
    \tilde{\lambda}-g_{\pi_{n'}(i'+1)}(\tilde{\lambda})
    = (1+C_{\pi_{n'}(i'+1),n',i'})\,\varepsilon > 0.
    \label{eq: false_positive_excluded}
\end{equation}
Therefore $\tilde{\lambda}>\iota_{\pi_{n'}(i'+1)}(\tilde{\alpha})$, so the pair is still rejected. Hence sufficiently small numerical error cannot turn a non-optimal pair into an accepted one.
\end{proof}

\end{document}